\documentclass[journal]{IEEEtran}
\makeatletter

\let\proof\@undefined
\let\endproof\@undefined
\makeatother
\usepackage{epsfig,subfigure,graphicx}
\usepackage{amsmath}
\usepackage{amsthm}
\usepackage{amssymb}
\usepackage[ruled,boxed]{algorithm}
\usepackage{algorithmic}
\usepackage{epstopdf}

\newtheorem{thm}{Theorem}[section]
\theoremstyle{definition}
\newtheorem{dfn}{Definition}[section]
\theoremstyle{remark}

\theoremstyle{plain}

\theoremstyle{remark}
\newtheorem{rmk}{Remark}[section]
\theoremstyle{plain}

\theoremstyle{plain}
\newtheorem{lem}{Lemma}[section]

\newcommand{\vect}[1]{\pmb{#1}}
\newcommand{\mat}[1]{\pmb{#1}}

\begin{document}
%
\title{Fast and Efficient Compressive Sensing using Structurally Random Matrices}
%
%
\author{Thong T. Do,~\IEEEmembership{Student Member,~IEEE,}
        Lu Gan,~\IEEEmembership{Member,~IEEE,} Nam H. Nguyen
        and Trac D. Tran,~\IEEEmembership{Senior Member,~IEEE}
\thanks{This work has been supported in part by the
National Science Foundation under Grant CCF-0728893.}
\thanks{Thong T. Do, Nam Nguyen and Trac D. Tran are with the Johns Hopkins University, Baltimore, MD, 21218 USA.}
\thanks{Lu Gan is with the Brunel University, London, UK.}}
%
%
%
\markboth{IEEE TRANSACTIONS ON SIGNAL PROCESSING,~Vol.~XXX,
No.~XXX,~XXX~2011}{Shell \MakeLowercase{\textit{et al.}}: Bare Demo
of IEEEtran.cls for Journals}



\maketitle

\begin{abstract}
This paper introduces a new framework of fast and efficient sensing
matrices for practical compressive sensing, called Structurally
Random Matrix (SRM). In the proposed framework, we pre-randomize a
sensing signal by scrambling its samples or flipping its sample
signs and then fast-transform the randomized samples and finally,
subsample the transform coefficients as the final sensing
measurements. SRM is highly relevant for large-scale, real-time
compressive sensing applications as it has fast computation and
supports block-based processing. In addition, we can show that SRM
has theoretical sensing performance comparable with that of
completely random sensing matrices. Numerical simulation results
verify the validity of the theory as well as illustrate the
promising potentials of the proposed sensing framework.
\end{abstract}


\section{Introduction}\label{intro_sec}

\label{sec:intro} \PARstart{C}{omp}ressed sensing (CS)
\cite{CandesRob06,DonohoCom06} has attracted a lot of interests over
the past few years as a revolutionary signal sampling paradigm.
Suppose that $\vect x$ is a length-$N$ signal. It is said to be
$K$-sparse (or compressible) if $\vect x$ can be well approximated
using only $K\ll N$ coefficients under some linear transform:
\begin{equation*}
\vect x=\mat \Psi \vect \alpha,
\end{equation*}
where $\mat \Psi$ is the sparsifying basis and $\vect \alpha$ is the
transform coefficient vector that has $K$ (significant) nonzero
entries.

According to the CS theory, such a signal can be acquired through
the following random linear projection:
\begin{equation*}\label{eq:CS-basic}
\vect y=\mat \Phi \vect x+ \vect e,
\end{equation*}
where $\vect y$ is the sampled vector with $M\ll N$ data points,
$\mat \Phi$ represents a $M\times N$ random matrix and $\vect e$ is
the acquisition noise. The CS framework is attractive as it implies
that $\vect x$ can be faithfully recovered from only
$M=\mathcal{O}(K\log N)$ measurements, suggesting the potential of
significant cost reduction in digital data acquisition.

While the sampling process is simply a random linear projection, the
reconstruction to find the sparsest signal from the received
measurements is highly non-linear process. More precisely, the
reconstruction algorithm is to solve the $l_1$-minimization of a
transform coefficient vector:
\begin{equation*}
\min \|\vect \alpha\|_1 \hspace{0.5cm}\text{s.t.}
\hspace{0.5cm}\vect y =\mat \Phi \mat \Psi \vect \alpha.
\end{equation*}

Linear programming \cite{CandesRob06, DonohoCom06} and other convex
optimization algorithms \cite{MarioGrad07, ElaineFix07,
EwoutProbing08} have been proposed to solve the $l_1$ minimization.
Furthermore, there also exists a family of greedy pursuit algorithms
\cite{TroppSig05, NeedellCosamp08, WeiSub08, DonohoSpa06,
DoSparity08} offering another promising option for sparse
reconstruction. These algorithms all need to compute $\mat \Phi\mat
\Psi$ and $(\mat \Phi\mat \Psi)^T$ multiple times. Thus,
computational complexity of the system depends on the structure of
sensing matrix $\mat \Phi$ and its transpose $\mat \Phi^T$.

Preferably, the sensing matrix $\mat \Phi$ should be highly
incoherent with sparsifying basis $\mat \Psi$, i.e. rows of $\mat
\Phi$ do not have any sparse representation in the basis $\mat
\Psi$. Incoherence between two matrices is mathematically quantified
by the mutual coherence coefficient \cite{DonohoUnc01}.
\begin{dfn}
The mutual coherence of an orthonormal matrix $N\times N$ $\mat
\Phi$ and another orthonormal matrix $N\times N$ $\mat \Psi$ is
defined as:
\begin{equation*}
\mu(\mat \Phi, \mat \Psi) = \max_{1\leq i,j\leq N} |\langle \vect
\Phi_i, \vect \Psi_j \rangle|
\end{equation*}
where $\vect \Phi_i$ are rows of $\mat \Phi$ and $\vect \Psi_j$ are
columns of $\mat \Psi$, respectively.
\end{dfn}

If $\mat \Phi$ and $\mat \Psi$ are two orthonormal matrices, $\|\mat
\Phi\vect \Psi_j\|_2 = \|\vect \Psi_j\|_2 = 1$. Thus, it is easy to
see that for two orthonormal matrices $\mat \Phi$ and $\mat \Psi$ ,
$1/\sqrt{N}\leq \mu \leq 1$. Incoherence implies that the mutual
coherence or the maximum magnitude of entries of the product matrix
$\mat \Phi \mat \Psi$ is relatively small. Two matrices are
completely incoherent if their mutual coherence coefficient
approaches the lower bound value of $1/\sqrt{N}$.

A popular family of sensing matrices is a random projection or a
random matrix of i.i.d random variables from a sub-Gaussian
distribution such as Gaussian or Bernoulli \cite{CandesNea06,
MedelsonUni06}. This family of sensing matrix is well-known as it is
universally incoherent with all other sparsifying basis. For
example, if $\mat \Phi$ is a random matrix of Gaussian i.i.d entries
and $\mat \Psi$ is an arbitrary orthonormal sparsifying basis, the
sensing matrix in the transform domain $\mat \Phi \mat \Psi$ is also
Gaussian i.i.d matrix. The universal property of a sensing matrix is
important because it enables us to sense a signal directly in its
original domain without significant loss of sensing efficiency and
without any other prior knowledge. In addition, it can be shown that
random projection approaches the optimal sensing performance of $M =
\mathcal{O}(K\log N)$.

However, it is quite costly to realize random matrices in practical
sensing applications as they require very high computational
complexity and huge memory buffering due to their completely
unstructured nature \cite{CandesSpa07}. For example, to process a
$512\times 512$ image with $64K$ measurements (i.e., $25\%$ of the
original sampling rate), a Bernoulli random matrix requires nearly
gigabytes storage and giga-flop operations, which makes both the
sampling and recovery processes very expensive and in many cases,
unrealistic.

Another class of sensing matrices is a uniformly random subset of
rows of an orthonormal matrix in which the partial Fourier matrix
(or the partial FFT) is a special case \cite{MedelsonUni06,
CandesSpa07}. While the partial FFT is well known for having fast
and efficient implementation, it only works well in the transform
domain or in the case that the sparsifying basis is the identity
matrix. More specifically, it is shown in [\cite{CandesSpa07},
\textit{Theorem} $1.1$] that the minimal number of measurements
required for exact recovery depends on the incoherence of $\mat
\Phi$ and $\mat \Psi$:
\begin{equation}\label{measnum}
M = \mathcal{O}(\mu_n^2 K \log N)
\end{equation}
where $\mu_n$ is the normalized mutual coherence: $\mu_n =
\sqrt{N}\mu$ and $1\leq \mu_n \leq \sqrt{N}$. With many well-known
sparsifying basis such as wavelets, this mutual coherence
coefficient might be large and thus, resulting in performance loss.
Another approach is to design a sensing matrix to be incoherent with
a given sparsifying basis. For example, \textit{Noiselets} is
designed to be incoherent with the Haar wavelet basis in
\cite{CofimanNoi05}, i.e. $\mu_n=1$ when $\mat \Phi$ is Noiselets
transform and $\mat \Psi$ is the Haar wavelet basis. Noiselets also
has low-complexity implementation $\mathcal{O}(N\log N)$ although it
is unknown if noiselets is also incoherent with other bases.

\section{Compressive Sensing with Structurally Random Matrices }\label{intro_srm_sec}
\subsection{Overview}

One of remaining challenges for CS in practice is to design a CS
framework that has the following features:
\begin{itemize}
\item \textit{Optimal or near optimal sensing performance}: the number of
measurements for exact recovery approaches the minimal bound, i.e.
on the order of $\mathcal{O}(K \log N)$;
\item \textit{Universality}: sensing performance is equally good with almost all sparsifying bases;
\item \textit{Low complexity, fast computation and block-based processing support}: these features of the sensing matrix are desired for large-scale, realtime sensing
applications;
\item \textit{Hardware/Optics implementation friendliness}: entries of the sensing matrix
only take values in the set $\{0,1,-1\}$.
\end{itemize}

In this paper, we propose a framework that aims to satisfy the above
wish-list, called \textit{Structurally Random Matrix}(SRM) that is
defined as a product of three matrices:
\begin{equation}
\mat \Phi = \sqrt{\frac{N}{M}}\mat D \mat F \mat R
\end{equation}
where:
\begin{itemize}
\item  $\mat R \in N\times N$ is either a uniform random permutation
matrix or a diagonal random matrix whose diagonal entries $R_{ii}$
are i.i.d Bernoulli random variables with identical distribution
$P(R_{ii}=\pm 1)=1/2$. A uniformly random permutation matrix
scrambles signal's sample locations globally while a diagonal matrix
of Bernoulli random variables flips signal's sample signs locally.
Hence, we often refer the former as the \textit{global randomizer}
and the latter as the \textit{local randomizer}.

\item $\mat F \in N\times N$ is an orthonormal matrix that,in practice, is selected to be fast computable such as popular fast transforms: FFT, DCT, WHT or their block
diagonal versions. The purpose of the matrix $\mat F$ is to spread
\textit{information} (or energy) of the signal's samples over all
measurements
\item  $\mat D \in M\times N$ is a subsampling matrix/operator. The operator $\mat D$ selects a random subset of rows of the matrix $\mat F \mat R$. If the probability of selecting a row $P(\text{a row is selected})$ is
$M/N$, the number of rows selected would be $M$ in average. In
matrix representation, $\mat D$ is simply a random subset of $M$
rows of the identity matrix of size $N\times N$. The scale
coefficient $\sqrt{\frac{N}{M}}$ is to normalize the transform so
that energy of the measurement vector is almost similar to that of
the input signal vector.
\end{itemize}

Equivalently, the proposed sensing algorithm SRM contains 3 steps:
\begin{itemize}
\item Step 1 (Pre-randomize): Randomize a target signal by either flipping
its sample signs or uniformly permuting its sample locations. This
step corresponds to multiplying the signal with the matrix $\mat R$
\item Step 2 (Transform): Apply a fast transform $\mat F$ to the randomized signal
\item Step 3 (Subsample): randomly pick up $M$ measurements out of N
transform coefficients. This step corresponds to multiplying the
transform coefficients with the matrix $\mat D$
\end{itemize}

Conventional CS reconstruction algorithm is employed to recover the
transform coefficient vector $\vect \alpha$ by solving the $l_1$
minimization:
\begin{equation}
\widehat{\vect \alpha}=\text{argmin} \|\vect \alpha\|_1
\hspace{0.4cm} s.t. \hspace{0.4cm}\vect y=\mat \Phi \mat \Psi \vect
\alpha.
\end{equation}

Finally, the signal is recovered as $\widehat{\vect x}=\mat \Psi
\widehat{\vect \alpha}$. The framework can achieve perfect
reconstruction if $\widehat{\vect x}=\vect x$.

From the best of our knowledge, the proposed sensing algorithm is
distinct from currently existing methods such as random projection
\cite{CandesDec05}, random filters \cite{TroppRan06}, structured
Toeplitz \cite{WaheedToeplitzCS07} and random convolution
\cite{JustinRandomconvol08} via the first step of pre-randomization.
Its main purpose is to scramble the structure of the signal,
converting the sensing signal into a white noise-like one to achieve
universally incoherent sensing.

Depending on specific applications, SRM can offer computational
benefits either at the sensing process or at the signal
reconstruction process. For applications that allow us to perform
sensing operation by computing the complete transform $\mat F$, we
can exploit the fast computation of the matrix $\mat F$ at the
sensing side. However, if it is required to precompute $\mat D\mat
F\mat R$ (and then store it in the memory for future sensing
operation), there would not be any computational benefit at the
sensing side. In this case, we can still exploit the structure of
SRM to speed up the signal recovery at the reconstruction side as in
most $l_1$-minimization algorithms \cite{MarioGrad07}, majority of
computational complexity is spent to compute matrix-vector
multiplications $\mat A \vect u$ and $\mat A^T\vect u$, where $\mat
A = \mat \Phi\mat \Psi$. Note that both $\mat A$ and $\mat A^T$ are
fast computable if the sparsifying matrix $\mat \Psi$ is fast
computable, i.e. their computational complexity on the order of
$\mathcal{O}(N\log N)$. In addition, when $\mat F$ is selected to be
the Walsh-Hadamard matrix, the SRM entries only take values in the
set $\{-1,1\}$, which is friendly for hardware/optics
implementation.

The remaining of the paper is organized as follows. We first discuss
about incoherence between SRMs and sparsifying transforms in
Section~\ref{incoherence_analysis_sec}. More specifically,
Section~\ref{asymp_distribution_analysis} will give us a rough
intuition of why SRM has sensing performance comparable with
Gaussian random matrices. Detail quantitative analysis of the
incoherence for SRMs with the local randomizer and the global
randomizer is presented in Section~\ref{incoherence_analysis}. Based
on these incoherence results, theoretical performance of the
proposed framework  is analyzed in Section~\ref{CS_analysis_sec} and
then followed by experiment validation in
Section~\ref{numerical_experiment}. Finally,
Section~\ref{discussion_conclusion} concludes the paper with detail
discussion of practical advantages of the proposed framework and
relationship between the proposed framework and other related works.

\subsection{Notations}\label{notations}
We reserve a bold letter for a vector, a capital and bold letter for
a matrix, a capital and bold letter with one sub-index for a row or
a column of a matrix and a capital letter with two sub-indices for
an entry of a matrix. We often employ $\vect x \in \mathbb{R}^N$ for
the input signal, $\vect y \in \mathbb{R}^M$ for the measurement
vector, $\mat \Phi \in \mathbb{R}^{M\times N}$ for the sensing
matrix, $\mat \Psi \in \mathbb{R}^{N\times N}$ for the sparsifying
matrix and $\vect \alpha \in \mathbb{R}^N$ for the transform
coefficient vector ($\vect x = \mat \Psi \vect \alpha$). We use the
notation $\text{supp}(\vect z)$ to indicate the index set (or
coordinate set) of nonzero entries of the vector $\vect z$.
Occasionally, we also use $\mathcal{T}$ to alternatively refer to
this index set of nonzero entries (i.e., $\mathcal{T}$=supp($\vect
z$)). In this case, $\vect z_\mathcal{T}$ denotes the portion of
vector $\vect z$ indexed by the set $\mathcal{T}$ and $\mat
\Psi_\mathcal{T}$ denotes the submatrix of $\mat \Psi$ whose columns
are indexed by the set $\mathcal{T}$.

Let $\mat A = \mat F\mat R$ and $S_{ij}$, $F_{ij}$ be the entry at
the $i^{th}$ row and the $j^{th}$ column of $\mat A\mat \Psi$ and
$\mat F$, $R_{kk}$ be the $k^{th}$ entry on the diagonal of the
diagonal matrix $\mat R$, $\vect A_i$ and $\vect \Psi_j$ be the
$i^{th}$ row of $\mat A$ and $j^{th}$ column of $\mat \Psi$,
respectively.

In addition, we also employ the following notations:
\begin{itemize}
\item $x_n$ is on the order of $o(z_n)$, denoted as $x_n = o(z_n)$, if
\begin{equation*}
\lim_{n\rightarrow \infty}\frac{x_n}{z_n}= 0.
\end{equation*}
\item $x_n$ is on the order of $\mathcal{O}(z_n)$, denoted as $x_n = \mathcal{O}(z_n)$, if
\begin{equation*}
\lim_{n\rightarrow \infty}\frac{x_n}{z_n}= c.
\end{equation*}
where $c$ is some positive constant.
\item A random variable $X_n$ is called asymptotically normally
distributed $\mathcal{N}(0,\sigma^2)$, if
\begin{equation*}
\lim_{n\rightarrow \infty} P(\frac{X_n}{\sigma}\leq
x)=\frac{1}{\sqrt{2\pi}}\int_{-\infty}^x e^{\frac{-y^2}{2}} dy.
\end{equation*}
\end{itemize}

\section{Incoherence Analysis}\label{incoherence_analysis_sec}
\subsection{Asymptotical Distribution Analysis}\label{asymp_distribution_analysis}

If $\mat \Phi$ is an i.i.d Gaussian matrix
$\mathcal{N}(0,\frac{1}{N})$ and $\mat \Psi$ is an arbitrarily
orthonormal matrix, $\mat \Phi\mat \Psi$ is also i.i.d Gaussian
matrix $\mathcal{N}(0,\frac{1}{N})$, implying that with overwhelming
probability, a Gaussian matrix is highly incoherent with all
orthonormal $\mat \Psi$. In other words, the i.i.d. Gaussian matrix
is universally incoherent with fixed transforms (with overwhelming
probability). In this section, we will argue that under some mild
conditions, with $\mat \Phi = \mat D \mat F \mat R$, where $\mat D,
\mat F, \mat R$ are defined as in the previous section, entries of
$\mat \Phi\mat \Psi$ are asymptotically normally distributed
$\mathcal{N}(0,\sigma^2)$, where $\sigma^2 \leq
\mathcal{O}(\frac{1}{N})$. This claim is illustrated in
Fig.~\ref{normaldis}, which depicts the quantile-quantile (QQ) plots
of entries of $\mat \Phi\mat \Psi$, where $N=256$, $\mat F$ is the
$256\times 256$ DCT matrix and $\mat \Psi$ is the Daubechies-8
orthogonal wavelet basis. Fig.~\ref{normaldis}(a) and
Fig.~\ref{normaldis}(b) correspond to the case $\mat R$ is the local
and global randomizer, respectively. In both cases, the QQ-plots
appear straight, as the Gaussian model demands.

\begin{figure}[htb]
\centering%
\includegraphics[width=8.5cm]{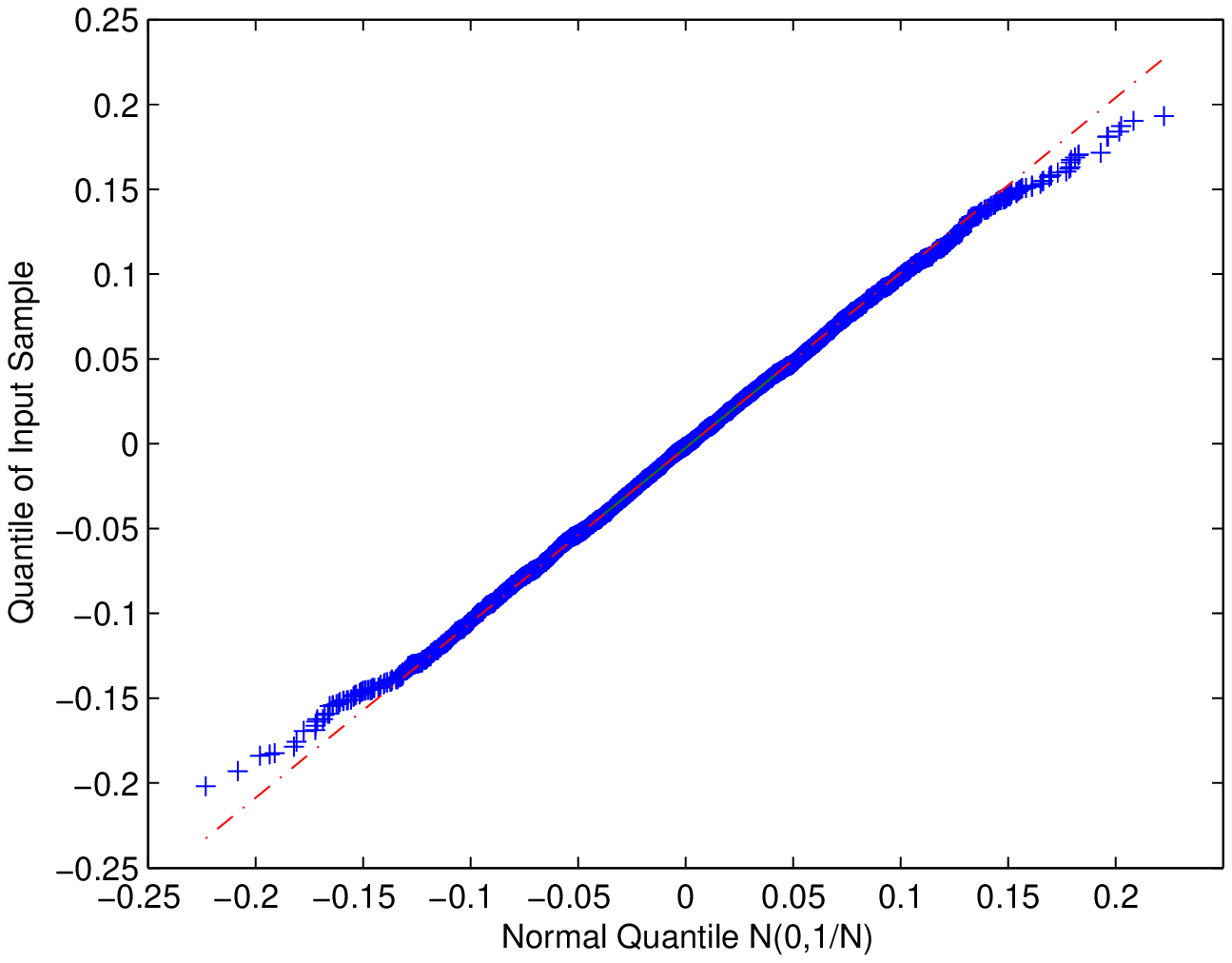}
\centerline{(a)}
\includegraphics[width=8.5cm]{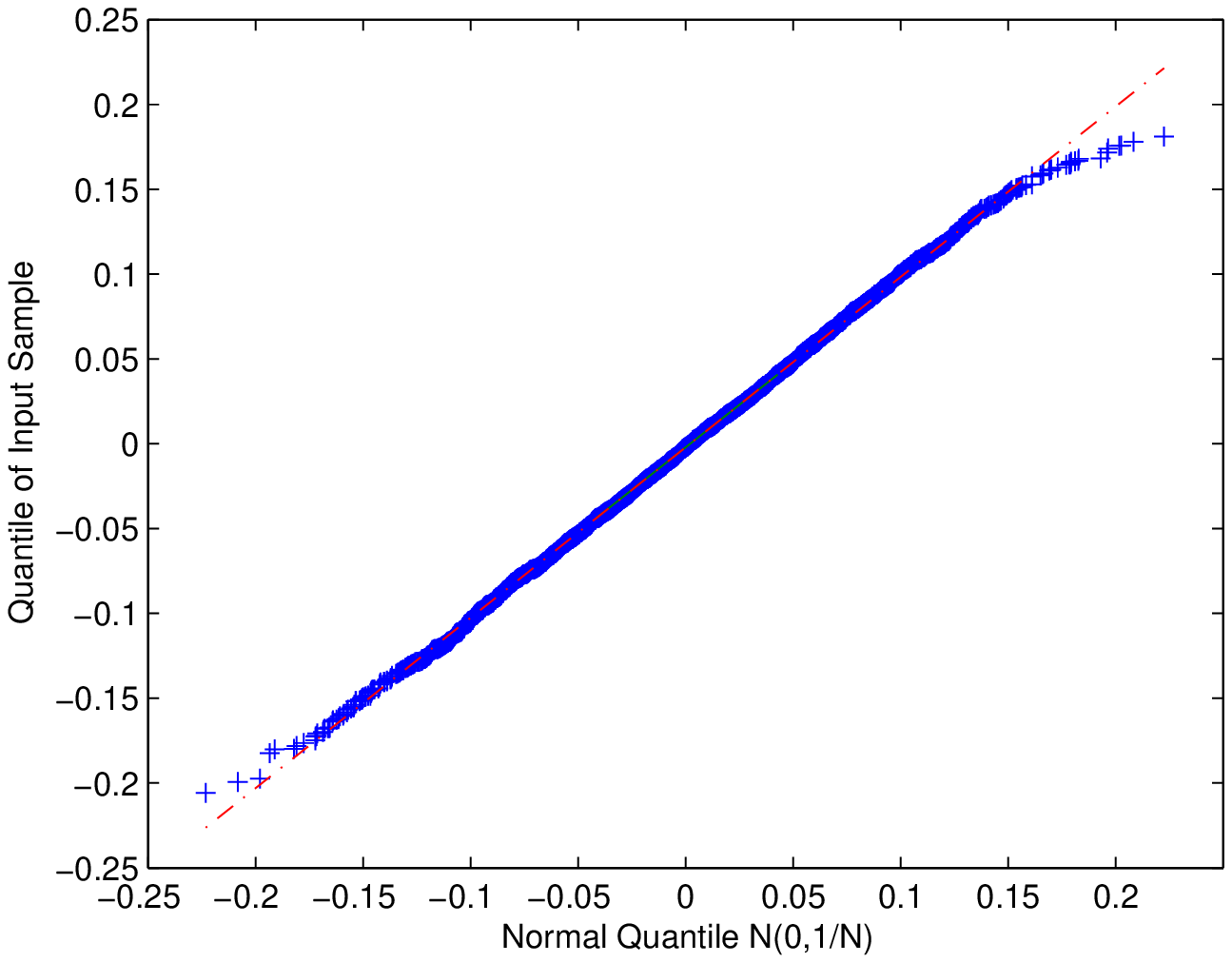}
\centerline{(b)} \caption{QQ plots comparing distribution of entries
of $\mat \Phi\mat \Psi$ and Gaussian distribution. (a) $\mat R$ is
the local randomizer. (b) $\mat R$ is the global randomizer. The
plots all appear nearly linear, indicating that entries of $\mat
\Phi\mat \Psi$ are nearly Normal distributed}\label{normaldis}
\end{figure}

Note that $\mat \Phi$ is a submatrix of $\mat A= \mat F \mat R$.
Thus, asymptotical distribution of the entries of $\mat A \mat \Psi$
is similar to that of entries of $\mat \Phi\mat \Psi$.

Before presenting the asymptotical theoretical analysis, we
introduce the following assumptions for the local and global
randomization models.

\subsubsection{Assumptions for the Local Randomization Model}
\begin{itemize}
\item $\mat F$ is an $N\times N$ unit-norm row matrix with absolute magnitude of all entries on the
order of $\mathcal{O}(\frac{1}{\sqrt{N}})$.
\item $\mat \Psi$ is an $N\times N$ unit-norm column matrix with the maximal absolute magnitude of
entries on the order of $o(1)$.
\end{itemize}
\subsubsection{Assumptions for the Global Randomization Model} The
global randomization model requires similar assumptions for the
local randomization model plus the following extra assumptions
\begin{itemize}
\item The average sum of entries on each column of $\mat \Psi$ is on the order of
$o(\frac{1}{\sqrt{N}})$.
\item Sum of entries on each row of $\mat F$ is zero.
\item Entries on each row of $\mat F$ and on each column of
$\mat \Psi$ are not all equal.
\end{itemize}

\begin{thm}\label{thm1}
Let $\mat A= \mat F \mat R$, where $\mat R$ is the local randomizer.
Given the assumptions for the local randomization model, entries of
$\mat A\mat \Psi$ are asymptotically normally distributed
$\mathcal{N}(0,\sigma^2)$ with $\sigma^2 \leq
\mathcal{O}(\frac{1}{N})$.
\end{thm}

\begin{proof}

With notations being defined in Section~\ref{notations}, we have:
\begin{equation}
S_{ij} = \langle \vect A_i, \vect \Psi_j\rangle = \sum_{k=1}^N
F_{ik} \Psi_{kj} R_{kk}
\end{equation}
Denote $Z_k = F_{ik} \Psi_{kj} R_{kk}$. Because $R_{kk}$ are i.i.d
Bernoulli random variables, $Z_k$ are i.i.d zero-mean random
variables with $E(Z_k)= 0$. The assumption that $|F_{ik}|$ are on
the order of $\mathcal{O}(\frac{1}{\sqrt{N}})$ implies that there
exist two positive constants $c_1$ and $c_2$ such that:

\begin{equation}\label{inequa1}
\frac{c_1}{N}\Psi_{kj}^2 \leq \text{Var}(Z_k) = F_{ik}^2\Psi_{kj}^2
\leq \frac{c_2}{N}\Psi_{kj}^2.
\end{equation}

\noindent The variance of $S_{ij}$, $\sigma^2$, can be bounded as
the follows:
\begin{equation}\label{inequa2}
\frac{c_1}{N} = \frac{c_1}{N}\sum_{k=1}^N \Psi_{kj}^2 \leq \sigma^2
= \sum_{k=1}^N \text{Var}(Z_k) \leq \frac{c_2}{N}\sum_{k=1}^N
\Psi_{kj}^2 = \frac{c_2}{N}.
\end{equation}

Because $S_{ij}$ is a sum of i.i.d zero-mean random variables
$\{Z_k\}_{k=1}^N$, according to the Central Limit Theorem (CLT)(see
Appendix~\ref{appdx1}), $S_{ij}\rightarrow
\mathcal{N}(0,\mathcal{O}(\frac{1}{N}))$. To apply CLT, we need to
verify its convergence condition: for a given $\epsilon >0$ and
there exists $N$ that is sufficiently large such that the
$\text{Var}(Z_k)$ satisfy:
\begin{equation}
\text{Var}(Z_k)<\epsilon \sigma^2, k=1,2,...,N.
\end{equation}

To show that this convergence condition is met, we use the
counterproof method. Assume there exists $\epsilon_0$ such that
$\forall N$, there exists at least $k_0\in \{1,2,\dots,N\}$:

\begin{equation}\label{inequa3}
\text{Var}(Z_{k_0})>\epsilon_0 \sigma^2.
\end{equation}

From (\ref{inequa1}), (\ref{inequa2}) and (\ref{inequa3}), we
achieve:
\begin{equation}
\epsilon_0 \frac{c_1}{N} \leq \text{Var}(Z_{k_0})\leq
\frac{c_2}{N}\Psi_{k_0j}^2.
\end{equation}

This inequality can not be true if $\Psi_{k_0j}$ is on the order of
$o(1)$. The underlying intuition of the convergence condition is to
guarantee that there is no random variable with dominant variance in
the sum $S_{ij}$. In this case, it simply requires that there is no
dominant entry on each column of $\mat \Psi$.
\end{proof}

Similarly, we can obtain a similar result when $\mat R$ is a
uniformly random permutation matrix.

\begin{thm}\label{thm2}
Let $\mat A= \mat F \mat R$, where $\mat R$ is the global
randomizer. Given the assumptions for the global randomization
model, entries of $\mat A\mat \Psi$ are asymptotically normally
distributed $\mathcal{N}(0,\sigma^2)$, where $\sigma^2\leq
\mathcal{O}(\frac{1}{N})$.
\end{thm}

\begin{proof}
Let $[\omega_1, \omega_2,...,\omega_N]$ be a uniform random
permutation of $[1,2,...,N]$. Note that $\{\omega_k\}_{k=1}^N$ can
be viewed as a sequence of random variables with identical
distribution. In particular, for a fixed $k$:
\begin{equation*}
P(\omega_k = i)=\frac{1}{N}, i=1,2,...,N.
\end{equation*}
Denote $Z_k = F_{i\omega_k}\Psi_{kj}$ (we omit the dependence of
$Z_k$ on $i$ and $j$ to simplify the notation), we have:
\begin{equation*}
S_{ij} = \langle \vect A_i, \vect \Psi_j\rangle = \sum_{k=1}^N
F_{i\omega_k} \Psi_{kj} = \sum_{k=1}^N Z_k.
\end{equation*}
Using the assumption that the vector $\vect F_i$ has zero average
sum and unit norm, we derive:
\begin{equation*}
E(Z_k) = \Psi_{kj}E(F_{i\omega_k}) = \frac{\Psi_{kj}}{N}\sum_{j=1}^N
F_{ij} = 0.
\end{equation*}
and also,
\begin{equation*}
E(Z_k^2) =\Psi_{kj}^2
E(F_{i\omega_k}^2)=\frac{\Psi_{kj}^2}{N}\sum_{j=1}^N F_{ij}^2 =
\frac{\Psi_{kj}^2}{N}.
\end{equation*}

In addition, note that although $\{\omega_k\}_{k=1}^N$ have the
identical distribution, they are correlated random variables because
of the uniformly random permutation \textit{without} replacement.
Thus, with a pair of $k$ and $l$ such that $1\leq k\neq l\leq N$, we
have:
\begin{equation*}
\begin{split}
E(Z_kZ_l) & =\Psi_{kj} \Psi_{lj} E(F_{i\omega_k}F_{i\omega_l})\\
          & =\frac{\Psi_{kj} \Psi_{lj}}{N(N-1)}\sum_{1\leq p\neq q\leq N} F_{ip}F_{iq}\\
          & =\frac{\Psi_{kj} \Psi_{lj}}{N(N-1)}((\sum_{p=1}^N F_{ip})^2 - \sum_{p=1}^N F_{ip}^2)\\
          & =-\frac{\Psi_{kj} \Psi_{lj}}{N(N-1)}.
\end{split}
\end{equation*}

The last equation holds because the vector $\vect F_i$ has zero
average sum and unit-norm. Then, we derive the expectation and the
variance of $S_{ij}$ as follows:
\begin{equation*}
E(S_{ij})=0;
\end{equation*}
\begin{equation*}
\begin{split}
\text{Var}(S_{ij})& = \sum_{k=1}^N E(Z_k^2)+ \sum_{1\leq k\neq q\leq N}E(Z_k Z_l)\\
           & = \frac{1}{N}\sum_{k=1}^N \Psi_{kj}^2 - \frac{1}{N(N-1)}\sum_{1\leq k\neq l\leq N}\Psi_{kj}\Psi_{lj}\\
           & = \frac{1}{N} - \frac{1}{N(N-1)}((\sum_{k=1}^N \Psi_{kj})^2 - \sum_{k=1}^N \Psi_{kj}^2)\\
           & =\frac{1}{N} - \frac{1}{N(N-1)}((\sum_{k=1}^N \Psi_{kj})^2 -1)\\
           & \leq \frac{1}{N} + \frac{1}{N(N-1)} =
           \mathcal{O}(\frac{1}{N}).
\end{split}
\end{equation*}

The forth equations holds because the column $\vect \Psi_j$ has
unit-norm. The theorem is then a simple corollary of the
Combinatorial Central Limit Theorem \cite{HoeffdingCom51} (see
Appendix $1$), provided that its convergence condition can be
verified that is:
\begin{equation}\label{CLTconvergence}
\lim_{N\rightarrow \infty} N\frac{\max_{1\leq k\leq
N}(F_{ik}-\overline{F_i})^2}{\sum_{k=1}^N
(F_{ik}-\overline{F_i})^2}\frac{\max_{1\leq k\leq N}
(\Psi_{kj}-\overline{\Psi_j})^2}{\sum_{k=1}^N
(\Psi_{kj}-\overline{\Psi_j})^2} = 0,
\end{equation}
where
\begin{equation*}
\overline{F_i} = \frac{1}{N}\sum_{k=1}^N F_{ik};\hspace{0.4cm}
\overline{\Psi_j} = \frac{1}{N}\sum_{k=1}^N \Psi_{kj}.
\end{equation*}

Because $\overline{F_i}=0$, $\|F_i\|_2^2=1$ and $\max_{1\leq k\leq
N}F_{ik}^2=\mathcal{O}(\frac{1}{N})$, the equation
(\ref{CLTconvergence}) holds if the following equation holds:
\begin{equation}\label{newcond1}
\lim_{N\rightarrow \infty}\frac{\max_{1\leq k\leq N}
(\Psi_{jk}-\overline{\Psi_j})^2}{\sum_{k=1}^N
(\Psi_{jk}-\overline{\Psi_j})^2} = 0.
\end{equation}
Because $\{|\overline{\Psi_j}|\}_{j=1}^N$ are on the order of
$o(\frac{1}{\sqrt{N}})$:
\begin{equation}\label{newcond2}
\sum_{k=1}^N (\Psi_{kj}-\overline{\Psi_j})^2 = \|\Psi_j\|_2^2
-N\overline{\Psi_j}^2 = 1 -N\overline{\Psi_j}^2  = \mathcal{O}(1).
\end{equation}
Also, due to $|\overline{\Psi_j}|\leq \max_{1\leq k\leq
N}|\Psi_{kj}|$ and $|\Psi_{kj}|$ are on the order of $o(1)$:
\begin{equation}\label{newcond3}
\max_{1\leq k\leq N} (\Psi_{kj}-\overline{\Psi_j})^2 \leq
4\max_{1\leq k\leq N}\Psi_{kj}^2 = o(1).
\end{equation}

Combination of (\ref{newcond2}) and (\ref{newcond3}) implies
(\ref{newcond1}) and thus the convergence condition of the
Combinatorial Central Limit Theorem is verified.
\end{proof}

The condition that each row of $\mat F$ has zero average sum is to
guarantee that entries of $\mat F\mat \Psi$ have zero mean while the
condition that entries on each row of $\mat F$ and on each column of
$\mat \Psi$ are not all equal is to prevent the degenerate case that
entries of $\mat F\mat \Psi$ might become a deterministic quantity.
For example, when entries of a row $\vect F_i$ are all equal
$\frac{1}{\sqrt{N}}$, $S_{ij} = \frac{1}{\sqrt{N}}\sum_{k=1}^N
\Psi_{kj}$, which is a deterministic quantity, not a random
variable. Note that these conditions are not needed when $\mat R$ is
the local randomizer.

If $\mat F$ is a DCT matrix, a (normalized) WHT matrix or a
(normalized) DFT matrix, all the rows (except for the first one)
have zero average sum due to the symmetry in these matrices. The
first row, whose entries are all equal $\frac{1}{\sqrt{N}}$, can be
considered as the averaging row, or a lowpass filtering operation.
When the input signal is zero-mean, this row might be chosen or not
without affecting quality of the reconstructed signal. Otherwise, it
should be included in the chosen row set to encode the signal's
mean. Lastly, the condition that absolute average sum of every
column of the sparsifying basis $\mat \Psi$ are on the order of
$o(\frac{1}{\sqrt{N}})$ is also close to the reality because the
majority of columns of the sparsifying basis $\mat \Psi$ can be
roughly viewed as bandpass and highpass filters whose average sum of
the coefficients are always zero. For example, if $\mat \Psi$ is a
wavelet basis (with at least one vanishing moment), then all columns
of $\mat \Psi$ (except one at DC) has column sum of zero.

The aforementioned theorems show that under certain conditions, the
majority of entries of $\mat A\mat \Psi$ (also $\mat \Phi\mat \Psi$)
behave like Gaussian random variables $\mathcal{N}(0,\sigma^2)$,
where $\sigma^2\leq \mathcal{O}(\frac{1}{N}$). Roughly speaking,
this behavior constitutes to a good sensing performance for the
proposed framework. However, these asymptotic results are not
sufficient for establishing sensing performance analysis because in
general, entries of $\mat A\mat \Psi$ are not stochastically
independent, violating a condition of a sensing Gaussian i.i.d
matrix. In fact, the sensing performance might be quantitatively
analyzed by employing a powerful analysis framework of a random
subset of rows of an orthonormal matrix \cite{CandesSpa07}. Note
that $\mat A$ is also an orthonormal matrix when $\mat R$ is the
local or the global randomizer.

Based on the Gaussian tail probability and a union bound for the
maximum absolute value of a random sequence, the maximum absolute
magnitude of $\mat A\mat \Psi$ can be asymptotically bounded as
follows:
\begin{equation*}
P(\max_{1\leq i,j\leq N}|S_{ij}|\geq t)\preceq 2N^2 \exp
(-\frac{t^2}{2\sigma^2})
\end{equation*}
where $\sigma^2\leq \frac{c}{N}$ and $c$ is some positive constant
and $\preceq$ stands for "asymptotically smaller or equal", i.e.,
when $N$ goes to infinity, $\preceq$ becomes $\leq$.

If we choose $t = \sqrt{\frac{2c\log (2N^2/\delta)}{N}}$, the above
inequality is equivalent to:
\begin{equation*}
P(\max_{1\leq i,j\leq N}|S_{ij}|\leq \sqrt{\frac{c\log
2(N/\delta)^2}{N}}) \succeq 1-\delta
\end{equation*}
which implies that with probability at least $1-\delta$, the mutual
coherence of $\mat A$ and $\mat \Psi$ is upper bounded by
$\mathcal{O}(\sqrt{\frac{\log (N/\delta)}{N}})$, which is close to
the optimal bound, except the $\log N$ factor.\\

In the following section, we will employ a more powerful tool from
the theory of concentration inequalities to analyze the coherence
between $\mat A = \mat F\mat R$ and $\mat \Psi$ when $N$ is finite.
We also consider a more general case that $\mat F$ is a sparse
matrix (e.g. a block-diagonal matrix).

\subsection{Incoherence Analysis}\label{incoherence_analysis}
Before presenting theoretical results for incoherence analysis, we
introduce assumptions for block-based local and global randomization
models.
\subsubsection{Assumptions for the Block-based Local Randomization Model}
\begin{itemize}
\item  $\mat F$ is an $N\times N$ unit-norm row matrix with the maximal absolute magnitude of entries
on the order of $\mathcal{O}(\frac{1}{\sqrt{B}})$, where $1\leq
B\leq N$, i.e. $\max_{1\leq i,j\leq N}|F_{ij}|= \frac{c}{\sqrt{B}}$,
where $c$ is some positive constant.
\item $\mat \Psi$ is an $N\times N$ unit-norm column matrix.
\end{itemize}
\subsubsection{Assumptions for the Block-based Global Randomization Model}
The block-based global randomization model requires similar
assumptions for the block-based local randomization model plus the
following assumption:
\begin{itemize}
\item All rows of $\mat F$ have zero average sum.
\end{itemize}

\begin{thm}\label{thm3}
Let $\mat A= \mat F \mat R$, where $\mat R$ is the local randomizer.
Given the assumptions for the block-based local randomization model,
then
\begin{itemize}
\item With probability at least $1-\delta$, the mutual coherence of
$\mat A$ and $\mat \Psi$ is upper bounded by
$\mathcal{O}(\sqrt{\frac{\log (N/\delta)}{B}})$.

\item In addition, if the maximal absolute magnitude of entries of
$\mat \Psi$ is on the order of $\mathcal{O}(\frac{1}{\sqrt{N}})$,
the mutual coherence is upper bounded by
$\mathcal{O}(\sqrt{\frac{\log (N/\delta)}{N}})$, which is
independent of $B$.
\end{itemize}
\end{thm}

\begin{proof}
A common proof strategy for this theorem as well as for other
theorems in this paper is to establish a large deviation inequality
that implies the quantity of our interest is concentrated around its
expected value with high probability. Proof steps include:
\begin{itemize}
\item Showing that the quantity of our interest is a sum of
independent random variables;
\item Bounding the expectation and variance of the quantity;
\item Applying a relevant concentration inequality of a sum of
random variables;
\item Applying a union bound for the maximum absolute value of a random sequence.
\end{itemize}

In this case, the quantity of interest is:
\begin{equation*}
S_{ij} = \langle \vect A_i, \vect \Psi_j\rangle = \sum_{k\in
\text{supp}(\vect F_i)} F_{ik} \Psi_{kj} R_{kk}
\end{equation*}
Denote $Z_k = F_{ik} \Psi_{kj} R_{kk}$, for $k\in \text{supp}(\vect
F_i)$ (in the support set of the row $\vect F_i$). Because $R_{kk}$
are i.i.d Bernoulli random variables, $Z_k$ are also i.i.d random
variables with $E(Z_k)= 0$. $Z_{kk}$ are also bounded because $Z_k =
\pm F_{ik} \Psi_{kj}$

$S_{ij}$ is a sum of independent, bounded random variables. Applying
the Hoeffding's inequality (see Appendix 2) yields:
\begin{equation*}
\Pr(|S_{ij}|\geq t)\leq 2\exp (-\frac{t^2}{\sum_{k\in
\text{supp}(\vect f_i)}F_{ik}^2\Psi_{jk}^2}).
\end{equation*}
The next step is to evaluate $\sigma^2=\sum_{k\in supp(\vect
f_i)}F_{ik}^2\Psi_{jk}^2$. Here, $\sigma^2$ can be roughly viewed as
the approximation of the variance of $S_{ij}$.

\begin{equation}\label{newcond5}
\sigma^2 \leq \max_{1\leq i,j\leq N}|F_{ij}|^2 \sum_{k\in
\text{supp}(\vect F_i)} \Psi_{kj}^2 \leq \max_{1\leq i,j\leq
N}|F_{ij}|^2 = \frac{c}{B}
\end{equation}

If the maximal absolute magnitude of entries of $\mat \Psi$ is on
the order of $\mathcal{O}(\frac{1}{\sqrt{N}})$:
\begin{equation*}
\max_{1\leq i,j\leq N}|\Psi_{ij}| = \frac{c}{\sqrt{N}},
\end{equation*}
where $c$ is some positive constant, then
\begin{equation}\label{newcond4}
\sigma^2 \leq \max_{1\leq i,j\leq N}|\Psi_{ij}|^2 \sum_{1\leq k\leq
N} F_{ik}^2 \leq \max_{1\leq i,j\leq N}|\Psi_{ij}|^2 = \frac{c}{N}.
\end{equation}

Finally, we derive an upper bound of the mutual coherence $\mu =
\max_{1\leq i,j\leq N} |S_{ij}|$ by taking a union bound for the
maximum absolute value of a random sequence:
\begin{equation*}
P(\max_{1\leq i,j\leq N}|S_ij|\geq t)\leq
2N^2\exp(\frac{-t^2}{\sigma^2}).
\end{equation*}
Choose $t=\sqrt{\sigma^2 \log (2N^2/\delta)}$, after simplifying the
inequality, we get:
\begin{equation*}
P(\max_{1\leq i,j\leq N}|S_{ij}|\leq \sqrt{\sigma^2 \log
(2N^2/\delta)})\geq 1-\delta.
\end{equation*}

Thus, with an arbitrarily $\mat \Psi$, (\ref{newcond5}) holds and we
achieve the first claim of the Theorem:
\begin{equation*}
P(\max_{1\leq i,j\leq N}|S_{ij}|\leq \sqrt{\frac{c\log
(2N^2/\delta)}{B}})\geq 1-\delta.
\end{equation*}

In the case that (\ref{newcond4}) holds, we achieve the second claim
of the Theorem:
\begin{equation*}
P(\max_{1\leq i,j\leq N}|S_{ij}|\leq \sqrt{\frac{c\log
(2N^2/\delta)}{N}})\geq 1-\delta.
\end{equation*}
\end{proof}

\begin{rmk}
When $\mat A$ is some popular transform such as the DCT or the
normalized WHT, the maximal absolute magnitude of entries is on the
order of $\mathcal{O}(\frac{1}{\sqrt{N}})$. As a result, the mutual
coherence of $\mat A$ and an \textit{arbitrary} $\mat \Psi$ is upper
bounded by $\mathcal{O}(\sqrt{\frac{\log (N/\delta)}{N}})$, which is
also consistent with our asymptotic analysis above. In other words,
when at least $\mat \Phi$ or $\mat \Psi$ is a \textit{dense and
uniform} matrix, i.e. the maximal absolute magnitude of their
entries is on the order of $\mathcal{O}(\frac{1}{\sqrt{N}})$, their
mutual coherence approaches the minimal bound, except for the $\log
N$ factor. In general, the mutual coherence between an arbitrary
$\mat \Psi$ and a sparse matrix $\mat A$ (e.g. block diagonal matrix
of block size $B$) might be $\sqrt{\frac{N}{B}}$ times larger.
\end{rmk}

Cumulative coherence is another way to quantify incoherence between
two matrices \cite{SchnassAve07}.
\begin{dfn}
The cumulative coherence of an $N\times N$ matrix $\mat A$ and an
$N\times K$ matrix $\mat B$ is defined as:
\begin{equation*}
\mu_c(\mat A,\mat B) =\max_{1\leq i\leq N}\sqrt{\sum_{1\leq j\leq
K}\langle \vect A_i, \vect B_j\rangle^2}
\end{equation*}
where $\vect A_i$ and $\vect B_j$ are rows of $\mat A$ and columns
of $\mat B$, respectively.
\end{dfn}

The cumulative coherence $\mu_c(\mat A,\mat B)$ measures the
\textit{average} incoherence between two matrices $\mat A$ and $\mat
B$ while mutual coherence $\mu(\mat A,\mat B)$ measures the
entry-wise incoherence. As a result, the cumulative coherence seems
to be a better indicator of average sensing performance. In many
cases, we are only interested in cumulative coherence between $\mat
A$ and $\mat \Psi_T$, where $T$ is the support of the transform
coefficient vector. As will be shown in the following section, the
cumulative coherence provides a more powerful tool to obtain a
tighter bound for the number of measurements required for exact
recovery.

From the definition of cumulative coherence, it is easy to verify
that $\mu_c\leq \sqrt{K} \mu$. If we directly apply the result of
the Theorem \ref{thm3}, we obtain a trivial bound of the cumulative
coherence: $\mu_c=\mathcal{O}(\sqrt{\frac{K\log N}{B}})$ for an
arbitrary basis $\mat \Psi$ and $\mu_c=\mathcal{O}(\sqrt{\frac{K\log
N}{N}})$ for a dense and uniform $\mat \Psi$. In fact, we can get
rid of the factor $\log N$ by directly measuring the cumulative
coherence from its definition.

\begin{thm}\label{thm4}
Let $\mat A= \mat F \mat R$, where $\mat R$ is the local randomizer.
Given the assumptions for the block-based local randomization model,
with probability at least $1-\delta$, the cumulative coherence of
$\mat A$ and $\mat \Psi_{\mathcal{T}}$, where $|\mathcal{T}|=K$, is
upper bounded by $\frac{2c}{\sqrt{B}}\max(\sqrt{K},4\sqrt{\log
(2N/\delta)})$.
\end{thm}
\begin{proof}
Denote $\mat U = \mat \Psi_{\mathcal{T}}^*$ and $\vect U_k$ are
columns of $\mat U$. Let $\vect A_i$ and $\vect \Psi_j$ ($j \in
\mathcal{T}$) be rows of $\mat A$ and columns of $\mat
\Psi_{\mathcal{T}}$, respectively.
\begin{equation*}
S_i = \sqrt{\sum_{j\in \mathcal{T}}\langle \vect A_i, \vect
\Psi_j\rangle^2} = \|\vect A_i\vect \Psi_\mathcal{T}\|_2 =
\|\sum_{k\in \text{supp}(\vect F_i)}R_{kk}F_{ik} \vect U_k\|_2.
\end{equation*}

Denote $\vect V_k =F_{ik}\vect U_k$ and $\mat V$ is the matrix of
columns $\vect V_k$, $k\in \text{supp}(\vect F_i)$. First, we derive
upper bound for the Frobenius norm of $\mat V$:
\begin{equation*}
\|\mat V\|_F^2 \leq \max_{1\leq i,j\leq N} F_{ij}^2 \|U\|_F^2 =
\frac{c^2K}{B}.
\end{equation*}

The last equation holds because $\|\mat U\|_F^2= K$. Also, the bound
for the spectral norm is:

\begin{equation*}
\begin{split}
\|V\|_2^2 & =\sup_{\|\vect \beta \|_2=1}\sum_{k\in \text{supp}(\vect F_i)}|\langle \vect \beta, \vect V_k \rangle|^2\\
          & =\sup_{\|\vect \beta \|_2=1}\sum_{k\in \text{supp}(\vect F_i)} F_{ik}^2(\sum_{j=1}^K\vect \beta_j U_{kj})^2\\
          & \leq \max_{1\leq i,j\leq N}F_{ij}^2 \sup_{\|\vect \beta \|_2=1}\sum_{1\leq k\leq N}|\langle \vect \beta, \vect U_k \rangle|^2\\
          & \leq \frac{c^2}{B}\|\mat U\|_2^2 = \frac{c^2}{B}.
\end{split}
\end{equation*}

The last equation holds because $\|\mat U\|_2^2 = 1$. Now, we have:
\begin{equation*}
S_i = \|\sum_{k\in \text{supp}(\vect F_i)}R_{kk}F_{ik}\vect U_k\|_2
=\|\sum_{k\in \text{supp}(\vect F_i)}R_{kk}\vect V_k\|_2.
\end{equation*}

Let us denote $\vect Z = \sum_{k\in \text{supp}(\vect
F_i)}R_{kk}\vect V_k$.

$\vect Z$ is a Rademacher sum of vectors and $S_i = \|\vect Z\|_2$
is a random variable. To show that $S_i$ is concentrated around its
expectation, we first derive bound of $E(\|\vect Z\|_2)$. It is easy
to verify that for a random variable $X$, $E(X)\leq \sqrt{E(X^2)}$.
Thus, we will derive the upper bound for the simpler quantity
$E(\|\vect Z\|_2^2)$
\begin{equation*}
\begin{split}
E(\|\vect Z\|_2^2) & = E(\vect Z^*\vect Z)=\sum_{k,l \in
\text{supp}(\vect F_i)}E(R_{kk}R_{ll})\langle \vect V_k, \vect V_l\rangle\\
             & = \sum_{k\in supp(\vect F_i)}\langle \vect V_k,\vect V_k\rangle =
             \|\mat V\|_F^2 = \frac{c^2K}{B}.
\end{split}
\end{equation*}

The third equality holds because $R_{kk}$ are i.i.d Bernoulli random
variables and thus, $E(R_{kk}R_{ll}) = 0$ $\forall k\neq l$. As a
result,
\begin{equation*}
E(S_i) = E(\|\vect Z\|_2) \leq c\sqrt{\frac{K}{B}}.
\end{equation*}

Applying Ledoux's concentration inequality of the norm of a
Rademacher sum of vectors \cite{LedouxConcentration01} (see Appendix
2). Noting that $\|\mat V\|_2^2$ can be viewed as the variance of
$S_i$, yields:
\begin{equation*}
\Pr(S_i\geq c\sqrt{\frac{K}{B}} + t)\leq 2\exp(-t^2\frac{B}{16c^2})
\end{equation*}

Finally, apply a union bound for the maximum absolute value of a
random process,we obtain:
\begin{equation*}
\Pr(\max_{1\leq i\leq N}S_i\geq c\sqrt{\frac{K}{B}} + t)\leq 2N
\exp(-t^2\frac{B}{16c^2}).
\end{equation*}

Choose $t = \frac{4c}{\sqrt{B}}\sqrt{\log (2N/\delta)}$, we get:
\begin{equation*}
\Pr(\max_{1\leq i\leq N}S_i\geq \frac{c}{\sqrt{B}}(\sqrt{K} +
4\sqrt{\log (2N/\delta)}))\leq \delta.
\end{equation*}
Finally, we derive:
\begin{equation*}
\Pr(\max_{1\leq i\leq N}S_i\geq \frac{2c}{\sqrt{B}}\max(\sqrt{K},
4\sqrt{\log (2N/\delta)}))\leq \delta.
\end{equation*}
\end{proof}

\begin{rmk}
When $K\geq 16\log (2N/\delta)$, the cumulative coherence is upper
bounded by $\mathcal{O}(\sqrt{\frac{K}{B}})$. When $K\leq 16\log
(2N/\delta)$, the upper bound of the cumulative coherence is
$\mathcal{O}(\sqrt{\frac{\log (N/\delta)}{B}})$, which is similar to
that of the mutual coherence in Theorem \ref{thm3}.
\end{rmk}

\begin{rmk}
When $\mat F$ is some popular transform such as the DCT or the
normalized WHT, the maximum absolute magnitude of entries is on the
order of $\mathcal{O}(\frac{1}{\sqrt{N}})$. As a result, the
cumulative coherence of $\mat A$ and any arbitrary $\mat
\Psi_\mathcal{T}$,where $|\mathcal{T}|=K$, is upper bounded by
$\mathcal{O}(\sqrt{\frac{K}{N}})$ if $K> 16\log
(\frac{2N}{\delta})$.
\end{rmk}

\begin{rmk}

The above theorem represents the worst-case analysis because $\mat
\Psi$ can be an arbitrary matrix (the worst case corresponds to the
case when $\mat \Psi$ is the identity matrix). When $\mat \Psi$ is
known to be dense and uniform, the upper bound of cumulative
coherence, according to the Theorem \ref{thm3} and the fact that
$\mu_c\leq \mu\sqrt{K}$, is $\mathcal{O}(\sqrt{\frac{K\log N}{N}})$,
which is, in general, better than $\mathcal{O}(\sqrt{\frac{K}{B}})$.
\end{rmk}

The asymptotical distribution analysis in Section
\ref{asymp_distribution_analysis} reveals a significant technical
difference required for two randomization models. With the local
randomizer, entries of $\mat A\mat \Psi$ are sums of
\textit{independent} random variables while with the global
randomizer, they are sums of \textit{dependent} random variables.
Stochastic dependence among random variables makes it much harder to
set up similar arguments of their sum's concentration. In this case,
we will show that the incoherence of $\mat A$ and $\mat \Psi$ might
depend on an extra quantity, the \textit{heterogeneity coefficient}
of the matrix $\mat \Psi$.
\begin{dfn}
Assume $\mat \Psi$ is an $N\times N$ matrix. Let $\mathcal{T}_k$ be
the support of the column $\vect \Psi_k$. Define:
\begin{equation}\label{dfn1}
\rho_k = \frac{\max_{1\leq i\leq N}
|\Psi_{ki}|}{\sqrt{\frac{1}{|\mathcal{T}_k|}\sum_{i\in\
T_k}\Psi_{ki}^2}}.
\end{equation}
The column-wise heterogeneity coefficient of the matrix $\mat \Psi$
is defined as:
\begin{equation}\label{dfn2}
\rho_{\mat \Psi} =  \max_{1\leq k\leq N}\rho_k.
\end{equation}
\end{dfn}

Obviously, $1\leq \rho_k\leq \sqrt{|\mathcal{T}_k|}$. $\rho_k$
illustrates the difference between the largest entry's magnitude and
the average energy of \textit{nonzero} entries. Roughly speaking, it
indicates heterogeneity of nonzero entries of the vector $\vect
\Psi_k$. If nonzero entries of a column $\vect \Psi_k$ are
homogeneous, i.e. they are on the same order of magnitude, $\rho_k$
is on the order of a constant. If all nonzero entries of a matrix
are homogeneous, the heterogeneity coefficient is also on the order
of a constant, $C_{\mat \Psi} = \mathcal{O}(1)$ and $\mat \Psi$ is
referred as a uniform matrix. Note that a uniform matrix is not
necessarily dense, for example, a block-diagonal matrix of DCT or
WHT blocks

The following theorem indicates that when the global randomizer is
employed, the mutual coherence between $\mat A$ and $\mat \Psi$ is
upper-bounded by $\mathcal{O}(\rho_{\mat \Psi}\sqrt{\frac{\log
(N/\delta)}{B}})$, where $B$ is the block size of $\mat \Phi$ and
$\mat \Psi$ is an arbitrarily matrix with the heterogeneity
coefficient $\rho_{\mat \Psi}$.

\begin{thm}\label{thm5}
Let $\mat A= \mat F \mat R$, where $\mat R$ is the global
randomizer. Assume that $\rho_k\geq 4\log (2N^2/\delta)$ $\forall
k\in \{1,2,\dots,N\}$, where $\rho_k$ is defined as in (\ref{dfn1}).
Given the assumptions for the block-based global randomization
model, then

\begin{itemize}
\item With probability at least $1-\delta$, the mutual coherence of
$\mat A$ and $\mat \Psi$ is upper-bounded by $\mathcal{O}(\rho_{\mat
\Psi}\sqrt{\frac{\log (N/\delta)}{B}})$, where $\rho_{\mat \Psi}$ is
defined as in (\ref{dfn2})

\item In addition, if $\mat \Psi$ is dense and uniform, i.e. the
maximum absolute magnitude of its entries is on the order of
$\mathcal{O}(\frac{1}{\sqrt{N}})$ and $B\geq 4\log (2N^2/\delta)$,
the mutual coherence is upper-bounded by
$\mathcal{O}(\sqrt{\frac{\log (N/\delta)}{N}})$, which is
independent of $B$.
\end{itemize}
\end{thm}

\begin{proof}
Let $[\omega_1, \omega_2,\dots ,\omega_N]$ be a uniformly random
permutation of $[1,2,\dots ,N]$.

\begin{equation*}
S_{ij} = \langle \vect A_i,\vect \Psi_j\rangle = \sum_{k=1}^N
F_{i\omega_k}\Psi_{jk}.
\end{equation*}

As in the proof of the Theorem \ref{thm2}, $\{\omega_k\}_{k=1}^N$
can be viewed as a sequence of dependent random variables with
identical distribution, i.e. for a fixed $k\in \{1,2,\dots,N\}$:
\begin{equation*}
P(\omega_k = i) = \frac{1}{N}, \hspace{0.4cm} i\in \{1, 2, \dots,
N\}.
\end{equation*}

The condition of $\mat F$ is equivalent to $\max_{1\leq i,j \leq
N}|F_{ij}|=\frac{c}{\sqrt{B}}$, where $c$ is some positive constant.
Define $\{q_{k\omega_k}\}_{k=1}^N$ as the follows:
\begin{equation*}
q_{k\omega_k} =
\begin{cases}
\frac{\sqrt{B|\mathcal{T}_k|}}{2c\rho_{\mat \Psi}}F_{i\omega_k}\Psi_{jk}+\frac{1}{2} &\text{if $\Psi_{jk}\neq 0$}\\
0 & \text{if $\Psi_{jk}= 0$}.
\end{cases}
\end{equation*}

It is easy to verify that $0\leq q_{k\omega_k}\leq 1$. Define $W_k$
as the sum of dependent random variables $q_{k\omega_k}$
\begin{equation*}
\begin{split}
 W_k = \sum_{k=1}^N q_{k\omega_k} &=  \frac{\sqrt{B|T_k|}}{2c\rho_{\Psi}}\sum_{k=1}^N F_{i\omega_k}\Psi_{jk}+ \frac{|T_k|}{2}\\
                            &= \frac{\sqrt{B|T_k|}}{2c\rho_{\Psi}}S_{ij} +
                            \frac{|T_k|}{2}.
\end{split}
\end{equation*}

Note that $\{F_{i\omega_k}\}_{k=1}^N$ are zero-mean random variables
because $\vect F_i$ has zero average sum. Thus, $E(S_{ij})=0$ and
$E(W_k) = \frac{|T_k|}{2}$. Then, applying the Sourav's theorem of
concentration inequality for a sum of \textit{dependent} random
variables \cite{SouravStein07} (see Appendix 2) results in:
\begin{equation*}
P\{\frac{\sqrt{B|\mathcal{T}_k|}}{2c\rho_{\mat \Psi}}|S_{ij}|\geq
\epsilon \}\leq
2\exp(-\frac{\epsilon^2}{2|\mathcal{T}_k|+2\epsilon}).\\
\end{equation*}

Denote $t = \frac{2c\rho_{\mat
\Psi}}{\sqrt{B|\mathcal{T}_k|}}\epsilon$. The above inequality is
equivalent to:
\begin{equation*}
P\{|S_{ij}|\geq t \}\leq
2\exp(-\frac{B|\mathcal{T}_k|}{4c^2\rho_{\Psi}^2}\frac{t^2}{2|\mathcal{T}_k|
+ \frac{t}{c\rho_{\mat \Psi}}\sqrt{B|\mathcal{T}_k|}}).
\end{equation*}

By choosing $t = 4c\rho_{\mat \Psi}\sqrt{\frac{1}{B}\log
(\frac{2N^2}{\delta})}$, we achieve:
\begin{equation*}
P\{|S_{ij}|\geq t\}\leq 2\exp(\frac{-4|\mathcal{T}_k|\log
(\frac{2N^2}{\delta})}{2|\mathcal{T}_k|+4\sqrt{|\mathcal{T}_k|\log
(\frac{2N^2}{\delta})}}).
\end{equation*}

If $|\mathcal{T}_k|\geq 4\log (\frac{2N^2}{\delta})$, the
denominator inside the exponent is smaller than $4|\mathcal{T}_k|$.
Thus,

\begin{equation*}
P\{|S_{ij}|\geq 2c\rho_{\mat \Psi}\sqrt{\frac{1}{B}\log
(\frac{2N^2}{\delta})} \}\leq 2\exp(-\log
(\frac{2N^2}{\delta}))=\frac{\delta}{N^2}.
\end{equation*}

Finally, after taking the union bound for the maximum absolute value
of a random sequence and simplifying the inequality, we obtain the
first claim of the Theorem:
\begin{equation*}
P\{\max_{1\leq i,j\leq N}|S_{ij}|\leq \mathcal{O}(\rho_{\mat
\Psi}\sqrt{\frac{\log (N/\delta)}{B}})\}\geq 1-\delta.
\end{equation*}

If $\mat \Psi$ is known to be dense and uniform, i.e. $\max_{1\leq
i,j\leq N}|\Psi_{ij}| = \frac{c_1}{\sqrt{N}}$, where $c_1$ is some
positive constant. We then define $\{q_{k\omega_k}\}_{k=1}^N$ as the
following:
\begin{equation*}
q_{k\omega_k} =
\begin{cases}
\frac{\sqrt{BN}}{2c c_1}F_{ik}\Psi_{j\omega_k}+\frac{1}{2} &\text{if $F_{ik}\neq 0$}\\
0             & \text{if $F_{ik}= 0$}.
\end{cases}
\end{equation*}

Note that $0\leq q_{k\omega_k}\leq 1$ and $E(q_{k\omega_k}) =
\frac{B}{2}$. Repeat the same arguments above, we have:
\begin{equation*}
P\{|S_{ij}|\geq t \}\leq 2\exp(-\frac{NB}{4c^2c_1^2}\frac{t^2}{2B +
\frac{t}{c c_1}\sqrt{NB}}).
\end{equation*}

Similarly, choose $t = 4cc_1\sqrt{\frac{1}{N}\log
(\frac{2N^2}{\delta})}$, we can derive:
\begin{equation*}
P\{|S_{ij}|\geq t\}\leq 2\exp(\frac{-4B\log
(\frac{2N^2}{\delta})}{2B+4\sqrt{B\log (\frac{2N^2}{\delta})}}).
\end{equation*}

If $B\geq 4\log (\frac{2N^2}{\delta})$, the denominator inside the
exponent is smaller than $4B$. Thus,
\begin{equation*}
P\{|S_{ij}|\geq 2cc_1\sqrt{\frac{1}{N}\log (\frac{2N^2}{\delta})}
\}\leq \frac{\delta}{N^2}.
\end{equation*}

After taking the union bound of the maximum absolute value of a
random sequence, we achieve the second claim of the Theorem.
\end{proof}

\begin{rmk}
The first part of theorem implies that when $\mat F$ is a dense and
uniform matrix (e.g. DCT or normalized WHT) and $\mat \Psi$ is a
uniform matrix (not necessarily dense), the mutual coherence closely
approaches the minimum bound $\mathcal{O}(\sqrt{\frac{\log
(N/\delta)}{N}})$. Although in this theorem, the mutual coherence
depends on the heterogeneity coefficient, one will see in the
experimental Section \ref{numerical_experiment} that this dependence
is almost negligible in practice.
\end{rmk}

As a consequence of this theorem, when at least $\mat A$ or $\mat
\Psi$ is dense and uniform, the mutual coherence of $\mat A$ and
$\mat \Psi$ is roughly on the order of $\mathcal{O}(\sqrt{\frac{\log
N}{N}})$, which is quite close to the minimal bound
$\frac{1}{\sqrt{N}}$, except for the $\log N$ factor. Otherwise, the
coherence linearly depends on the block size $B$ of $\mat F$ and is
on the order of $\mathcal{O}(\sqrt{\frac{\log N}{B}})$. As a matter
of fact, this bound is almost optimal because when $\mat \Psi$ is
the identity matrix, the mutual coherence is actually equal the
maximum absolute magnitude of entries of $\mat A$, which is on the
order of $\mathcal{O}(\frac{1}{\sqrt{B}})$.

\begin{rmk}
Although the theoretical results of the global randomizer seem to be
always weaker than those of the local randomizer, there are a few
practical motivations to study this global randomizer. Speech
scrambling has been used for a long time for secure voice
communication. Also, analog image/video scrambling have been
implemented for commercial security related applications such as
CCTV surveillance system. In addition, permutation does not change
the dynamic range of the sensing signal, i.e. no bit expansion in
implementation. The computation cost of random permutation is only
$\mathcal{O}(N)$, which is very easy to implement in software. From
a security perspective the operation of random permutation offers a
large key space than random sign flipping ($N!$ vs $2^N$). Also, as
will be shown in the numerical experiment section, with random
permutation, one can get highly sparse measurement matrix.
\end{rmk}

\section{Compressive Sampling Performance Analysis}\label{CS_analysis_sec}

Section \ref{incoherence_analysis_sec} demonstrates that under some
mild conditions, the matrix $\mat A$ and $\mat \Psi$ are highly
incoherent, implying that the matrix $\mat A\mat \Psi$ is almost
dense. When $\mat A\mat \Psi$ is dense, energy of nonzero transform
coefficients $\vect \alpha_T$ is distributed over all measurements.
Commonly speaking, this is good for signal recovery from a small
subset of measurements because if energy of some transform
coefficients were concentrated in few measurements that happens to
be bypassed in the sampling process, there is no hope for exact
signal recovery even when employing the most sophisticated
reconstruction method. This section shows that a random subset of
rows of the matrix $\mat A = \mat F \mat R$ yields almost optimal
measurement matrix $\mat \Phi$ for compressive sensing.

\subsection{Assumptions for Performance Analysis}

A signal $\vect x$ is assumed to be sparse in some sparsifying basis
$\mat \Psi$: $\vect x=\mat \Psi \vect \alpha$, where the vector of
transform coefficients $\vect \alpha$ has no more than $K$ nonzero
entries. The sign sequence of nonzero transform coefficients $\vect
\alpha_T$ which is denoted as $\vect z$, is assumed to be a random
vector of i.i.d Bernoulli random variables (i.e. $P(z_i=\pm 1) =
\frac{1}{2}$). Let $\vect y=\mat \Phi \vect x$ be the measurement
vector, where $\mat \Phi=\sqrt{\frac{N}{M}}\mat D \mat F \mat R$ is
a Structurally Random Matrix. Assumptions of the block-based local
randomization and of the block-based global randomization models
hold.

\subsection{Theoretical Results}
\begin{thm}\label{thm6}
With probability at least $1-\delta$, the proposed sensing framework
can recover $K$-sparse signals exactly if the number of measurements
$M \geq \mathcal{O}(\frac{N}{B}K\log^2 (\frac{N}{\delta}))$. If
$\mat F$ is a dense and uniform rather than block-diagonal(e.g. DCT
or normalized WHT matrix), the number of measurement needed is on
the order of $\mathcal{O}(K\log^2 (\frac{N}{\delta}))$.
\end{thm}
\begin{proof}
This is a simple corollary of the theorem of Cand\`{e}s et. al.
[\cite{CandesSpa07} \textit{Theorem} $1.1$] (\ref{measnum}) because
(i) $\mat A = \mat F \mat R$ is an orthonormal matrix, and (ii) our
incoherence results between $\mat A$ and $\mat \Psi$ in the Theorem
\ref{thm3} and Theorem \ref{thm5}.
\end{proof}

\begin{rmk}
If $\mat \Psi$ is dense and uniform, the number of measurements for
exact recovery is always $\mathcal{O}(K\log^2 (\frac{N}{\delta}))$
regardless of the block size $B$. This implies that we can use the
identity matrix for the transform $\mat F$ (B = 1). For example,
when the input signal is known to be spectrally sparse,
compressively sampling it in the time domain is as efficient as in
any other transform domain.
\end{rmk}

Compared with the framework that uses random projection, there is an
upscale factor of $\log N$ for the number of measurements for exact
recovery. In fact, by employing the bound of cumulative coherence,
we can eliminate this upscale factor and thus, successfully showing
optimal performance guarantee.

\begin{thm}\label{thm7}
Assume that the sparsity $K>16\log (\frac{2N}{\delta})$. With
probability at least $1-\delta$, the proposed framework employing
the local randomizer can reconstruct $K$-sparse signals exactly if
the number of measurements $M \geq \mathcal{O}(\frac{N}{B}K\log
(\frac{N}{\delta}))$.If $\mat F$ is a dense and uniform matrix (e.g.
DCT or normalized WHT), the minimal number of required measurements
is $M=\mathcal{O}(K\log (\frac{N}{\delta}))$.
\end{thm}

\begin{proof}
The proof is based on the result of cumulative coherence in the
Theorem \ref{thm4} and a modification of the proof framework of the
compressed sensing \cite{CandesSpa07}.

Denote $\mat U =\sqrt{\frac{N}{M}}\mat F \mat R \mat \Psi$, $\mat
U_\mathcal{T}=\sqrt{\frac{N}{M}}\mat F \mat R \mat
\Psi_\mathcal{T}$, $\mat U_{\Omega }=\sqrt{\frac{N}{M}}\mat D \mat F
\mat R \mat \Psi$ and $\mat U_{\Omega
\mathcal{T}}=\sqrt{\frac{N}{M}}\mat D \mat F \mat R \mat
\Psi_\mathcal{T}$, where the support $\Omega = \{k|\mat D_{kk}=1,
k=1,2,..,N\}$. Let $\vect v_k$, $k\in \{1,2,...,N\}$, be columns of
$\mat U_\mathcal{T}^*$. Denote $\mu_c = \max_{1\leq k \leq N}\|\vect
v_k\|_2$, where $\mu_c = \mu_c(\mat A,\mat \Psi_\mathcal{T})$ is the
cumulative coherence of $\mat A= \sqrt{\frac{N}{M}}\mat F \mat R$
and $\mat \Psi_\mathcal{T}$. According to the above incoherence
analysis, $\mu_c\leq \mathcal{O}(\sqrt{\frac{KN}{BM}})$. Also,
denote $\mu$ as the mutual coherence of $\mat A$ and $\mat
\Psi_\mathcal{T}$, $\mu\leq \mathcal{O}(\sqrt{\frac{N\log N}{BM}})$.

As indicated in \cite{CandesNea06, CandesSpa07}, to show $l_1$
minimization exact recovery, it is sufficient to verify the
\textit{Exact Recovery Principle}.
\newtheoremstyle{nonum}{}{}{\itshape}{}{\bfseries}{.}{ }{\thmname{#1}\thmnote{ (\mdseries #3)}}
\theoremstyle{nonum}
\newtheorem{ERP}{Exact Recovery Principle}
\begin{ERP}
With high probability, $|\vect \pi_k|< 1$ for all $k\in
\mathcal{T}^c$, where $\mathcal{T}^c$ is the complementary set of
the set $\mathcal{T}$ and $\vect \pi = \mat U^*_{\Omega}\mat
U_{\Omega \mathcal{T}} (\mat U^*_{\Omega \mathcal{T}}\mat U_{\Omega
\mathcal{T}})^{-1}\vect z$, where $\vect z$ is the sign vector of
nonzero transform coefficients $\vect \alpha_\mathcal{T}$.
\end{ERP}

Note that $\vect \pi_k = \langle \vect \nu_k (\mat U^*_{\Omega
\mathcal{T}}\mat U_{\Omega \mathcal{T}})^{-1}, \vect z \rangle$,
where $\vect \nu_k$ is the $k^{th}$ row of $\mat U^*_{\Omega}\mat
U_{\Omega \mathcal{T}}$, for some $k\in \mathcal{T}^c$. To establish
the Exact Recovery Principle, we will first derive following lemmas.
The first lemma is to bound the norm of $\vect \nu_k$.

\begin{lem} (Bound the norm of $\vect \nu_k$)\label{lemma1}
With high probability, $\|\vect \nu_k\|$ is on the order of
$\mathcal{O}(\mu_c)$:
\begin{equation*}
P(\|\vect \nu_k\|\geq \mu_c + a\overline{\sigma})\leq 3\exp(-\gamma
a^2),
\end{equation*}
where $\overline{\sigma}$, $\gamma$ and $a$ are some certain
numbers.
\end{lem}
\begin{proof}

Let $\vect U_k$ be columns of $\mat U$. For $k\in \mathcal{T}^c$:
\begin{equation*}
\vect \nu_k =\frac{1}{M}\sum_{i=1}^N D_{ii} U_{ik} \vect v_i =
\sum_{i=1}^N (D_{ii}-\frac{M}{N}) U_{ik} \vect v_i
\end{equation*}
where the second equality holds because $\sum_{i=1}^N U_{ik} \vect
v_i = \mat U_T^*\vect U_k = 0$ that results from the orthogonality
of columns of $\mat U$. Let $Z_i = (D_{ii}-\frac{M}{N})$. Because
$D_{ii}$ are i.i.d binary random variables with $P(D_{ii} =
1)=\frac{M}{N}$, $Z_i$ are zero mean i.i.d random variables and
$E(Z_i^2) = \frac{M}{N}(1-\frac{M}{N})$. Let $\mat H$ be the matrix
of columns $\vect h_i = U_{ik}\vect v_i$, $i\in \{1,2,\dots,N\}$ .
Then, $\vect \nu_k$ can be viewed as a random weighted sum of column
vectors $\vect h_i$:
\begin{equation*}
\vect \nu_k = \frac{1}{M}\sum_{i=1}^N Z_i\vect h_i
\end{equation*}
and $\|\vect \nu_k\|$ is a random variable. We have:
\begin{equation*}
E(\|\vect \nu_k\|^2) = \sum_{1\leq i,j\leq N}E(Z_i Z_j)\langle \vect
W_i, \vect h_j\rangle = \sum_{1\leq i\leq N} E(Z_i^2)\|\vect
W_i\|^2,
\end{equation*}
where the last equality holds due to $E(Z_i Z_j)=0$ if $i\neq j$.
Thus,
\begin{equation*}
\begin{split}
E(\|\vect \nu_k\|^2) & =\frac{M}{N}(1-\frac{M}{N})\sum_{1\leq i\leq
N}
U_{ik}^2\|\vect v_i\|^2\\
               & \leq \frac{M}{N}(1-\frac{M}{N})\mu_c^2 \sum_{1\leq i\leq N} U_{ik}^2 \leq
               \mu_c^2.
\end{split}
\end{equation*}
where the last inequality holds due to $\|\vect U_k\|^2 =
\frac{N}{M}$. This implies that $E(\|\vect \nu_k\|) \leq \mu_c$. To
show that $\|\vect \nu_k\|$ is concentrated around its mean, we use
the Talagrand's theorem of concentration inequality
\cite{TalagrandConcentration96}. First, we have:
\begin{equation*}
\begin{split}
\|\mat H\|_2^2 & = \sup_{\|\vect \beta\|=1}\sum_{i=1}^N |\langle \vect \beta, \vect h_i\rangle |^2 = \sup_{\|\vect \beta\|=1}\sum_{i=1}^N U_{ik}^2 |\langle \vect \beta, \vect v_i\rangle |^2\\
          & \leq \mu^2 \sup_{\|\vect \beta\|=1}\sum_{i=1}^N |\langle \vect \beta, \vect v_i\rangle|^2
          = \mu^2\|\mat U_T\|_2^2 = \frac{N}{M}\mu^2.
\end{split}
\end{equation*}
where the last equation holds because $\|\mat U_T\|_2^2 =
\frac{N}{M}$. Thus, we derive the upper bound of the variance
$\sigma^2$:
\begin{equation*}
\sigma^2 = E(Z_k^2)\|\mat H\|_2^2 \leq
\frac{M}{N}(1-\frac{M}{N})\frac{N}{M}\mu^2\leq \mu^2.
\end{equation*}

In addition, it is obvious that $|Z_k|\leq 1$ and thus
\begin{equation*}
B =\max_{1\leq i\leq N}\|\vect h_i\|_2 \leq \mu \mu_c.
\end{equation*}

The Talagrand's theorem \cite{TalagrandConcentration96} (see
Appendix 2) shows that:
\begin{equation*}
P(\|\vect \nu_k\|-E(\|\vect \nu_k\|)\geq t)\leq
3\exp(\frac{-t}{cB}\log (1+\frac{Bt}{\sigma^2+BE(\|\vect
\nu_k\|)})),
\end{equation*}
where $c$ is some positive constant. Replacing $E(\|\vect \nu_k\|)$,
$\sigma^2$ and $B$ by their upper bounds in the right-hand side, we
obtain:
\begin{equation*}
P(\|\vect \nu_k\|-E(\|\vect \nu_k\|)\geq t)\leq
3\exp(\frac{-t}{c\mu\mu_c}\log (1+\frac{\mu\mu_c
t}{\mu^2+\mu\mu_c^2})).
\end{equation*}

The next step is to simplify the right-hand side of the above
inequality by replacing the denominator inside the $\log$ by two
times the dominant term and note that $log(1+x)\geq \frac{x}{2}$
when $x\leq 1$. In particular, there are two cases:
\begin{itemize}
\item Case 1: $\mu\mu_c^2\geq \mu^2$ or
equivalently, $\mu_c^2\geq \mu$, denote $\overline{\sigma}^2 =
\mu\mu_c^2$ and $t = a\overline{\sigma}$ . If $\mu\mu_c t\leq
2\mu\mu_c^2$ or equivalently, $a \leq 2(1/\mu)^{\frac{1}{2}}$,
\begin{equation*}
P(\|\vect \nu_k\|-E(\|\vect \nu_k\|)\geq t)\leq 3\exp(-\gamma a^2).
\end{equation*}
\item Case 2: $\mu^2\geq \mu\mu_c^2$, denote $\overline{\sigma}^2=\mu^2$ and $t =
a\overline{\sigma}$. If $\mu\mu_c t\leq 2\mu^2$ or equivalently,
$a\leq 2/\mu_c$
\begin{equation*}
P(\|\vect \nu_k\|-E(\|\vect \nu_k\|)\geq t)\leq 3\exp(-\gamma a^2).
\end{equation*}
where $\gamma$ is some positive constant.
\end{itemize}
In conclusion, let $\overline{\sigma} =
\sqrt{\max(\mu\mu_c^2,\mu^2)}$. Then, for any $a \leq \min(2/\mu_c,
2/\sqrt{\mu})$:
\begin{equation}\label{newcond6}
P(\|\vect \nu_k\|\geq \mu_c + a\overline{\sigma})\leq 3\exp(-\gamma
a^2),
\end{equation}
where $\gamma$ is some positive constant.
\end{proof}

The second lemma is to bound the spectral norm of $\mat U^*_{\Omega
T}\mat U_{\Omega \mathcal{T}}$
\begin{lem}\label{lemma2} (Bound the spectral norm of $\mat U^*_{\Omega T}\mat U_{\Omega
\mathcal{T}}$)

With high probability, $\|\mat U^*_{\Omega \mathcal{T}}\mat
U_{\Omega \mathcal{T}}\| \geq \frac{1}{2}$
\end{lem}

\begin{proof}
The Theorem $1.2$ in \cite{CandesSpa07} shows that with probability
$1-\delta$, $\|\mat U^*_{\Omega \mathcal{T}}\mat U_{\Omega
\mathcal{T}}\| \geq \frac{1}{2}$ if $M\geq \mu_c^2\max(c_1\log K,
c_2\log (3/\delta))$, where $c_1$ and $c_2$ are some known positive
constants.

\end{proof}

And the third lemma is to bound the norm of $\vect w_k = \vect \nu_k
(\mat U^*_{\Omega \mathcal{T}}\mat U_{\Omega T})^{-1}$
\begin{lem}\label{lemma3}
(Bound the norm of $\vect w_k = \vect \nu_k (\mat U^*_{\Omega
\mathcal{T}}\mat U_{\Omega T})^{-1}$)

With high probability, $\|\vect w_k\|$ is on the order of
$\mathcal{O}(\mu_c)$:
\begin{equation}\label{newcond7}
P(\sup_{k\in \mathcal{T}^c}\|\vect w_k\|\geq
2\mu_c+2a\overline{\sigma})\leq 3N\exp(-\gamma a^2)+P(\|\mat
U_{\Omega \mathcal{T}}^*\mat U_{\Omega \mathcal{T}}\|\leq
\frac{1}{2})
\end{equation}
where $a$, $\gamma$ and $\overline{\sigma}$ are defined in the proof
of the Lemma~\ref{lemma1}.
\end{lem}

\begin{proof}

Let $\mathcal{A}$ be the event that $\{\|\mat U_{\Omega
\mathcal{T}}^*\mat U_{\Omega \mathcal{T}}\|\geq \frac{1}{2}\}$ or
equivalently, $\{\|(\mat U_{\Omega \mathcal{T}}^*\mat U_{\Omega
\mathcal{T}})^-1\|\leq 2\}$ and $\mathcal{B}$ be the event that
$\{\sup_{k\in \mathcal{T}^c}\|\vect \nu_k\|\leq
\mu_c+a\overline{\sigma}\}$. Note that
\begin{equation*}
\sup_{k\in \mathcal{T}^c}\|\vect w_k\|\leq \|(\mat U_{\Omega
\mathcal{T}}^*\mat U_{\Omega \mathcal{T}})^{-1}\|\sup_{k\in
\mathcal{T}^c}\|\vect \nu_k\|.
\end{equation*}
Thus,
\begin{equation*}
P(\sup_{k\in \mathcal{T}^c}\|\vect w_k\|\geq
2\mu_c+2a\overline{\sigma})\leq P(\overline{\mathcal{A}\cap
\mathcal{B}})\leq
P(\overline{\mathcal{A}})+P(\overline{\mathcal{B}}).
\end{equation*}
Note that $P(\overline{\mathcal{B}})\leq 3N\exp(-\gamma a^2)$
implies (\ref{newcond7}) holds.

\end{proof}

To establish the Exact Recovery Principle, we will show that
$\sup_{k\in \mathcal{T}^c}|\langle \vect w_k, \vect z\rangle|\leq 1$
with high probability. Note that because $\vect z$ is assumed to be
a vector of i.i.d Bernoulli random variables, $|\langle \vect w_k,
\vect z\rangle|$ is concentrated around its zero mean. In
particular, according to the Hoeffding's inequality:
\begin{equation*}
P(|\langle \vect w_k, \vect z\rangle|\geq 1)\leq
2\exp(-\frac{1}{2\|\vect w_k\|^2}).
\end{equation*}
\begin{equation*}
\Rightarrow P(|\langle \vect w_k, \vect z\rangle|\geq 1|\sup_{k\in
\mathcal{T}^c}\|\vect w_k\|\leq \lambda)\leq
2N\exp(-\frac{1}{2\lambda^2}).
\end{equation*}
Note that with two arbitrary probabilistic events $\mathcal{A}$ and
$\mathcal{B}$:
\begin{equation*}
P(\mathcal{A}) =
P(\mathcal{A}|\mathcal{B})P(\mathcal{B})+P(\mathcal{A}|\overline{\mathcal{B}})P(\overline{\mathcal{B}})\leq
P(\mathcal{A}|\mathcal{B})+P(\overline{\mathcal{B}}).
\end{equation*}
Now, let $\mathcal{A}$ be the event $\{\sup_{k\in
\mathcal{T}^c}|\langle \vect w_k, \vect z\rangle|\geq 1\}$ and
$\mathcal{B}$ be the event $\{\sup_{k\in \mathcal{T}^c}\|\vect
w_k\|\leq \lambda \}$, we derive
\begin{equation}\label{newcond8}
P(\sup_{k\in \mathcal{T}^c}|\langle \vect w_k, \vect z\rangle|\geq
1)\leq 2N\exp(-\frac{1}{2\lambda^2})+P(\sup_{k\in
\mathcal{T}^c}\|\vect w_k\|\geq \lambda).
\end{equation}
Choose $\lambda = 2\mu_c+2a\overline{\sigma}$, according to
(\ref{newcond7}) and (\ref{newcond8}), the probability of our
interest $P(\sup_{k\in \mathcal{T}^c}|\langle \vect w_k, \vect
z\rangle|\geq 1)$ is upper bounded by:
\begin{equation*}
3N\exp(-\gamma a^2)+2N\exp(-\frac{1}{2\lambda^2})+\delta.
\end{equation*}

To show that $\{\sup_{k\in \mathcal{T}^c}|\langle \vect w_k, \vect
z\rangle|\leq 1\}$ with probability $1-\mathcal{O}(\delta)$, it is
sufficient to show that the above upper bound is not greater than
$3\delta$. In particular, choose $a^2 = \gamma^{-1}\log (3N/\delta)$
that makes the first term to be equal $\delta$.

To make the second term less than $\delta$, it is required that
\begin{equation}\label{newcond9}
\frac{1}{2\lambda^2}\geq \log (\frac{2N}{\delta}).
\end{equation}

\begin{itemize}
\item Case 1: $\mu_c^2\geq \mu$. The condition that (\ref{newcond6}) holds is $a \leq
2(1/\mu)^{\frac{1}{2}}$ that is equivalent to:
\begin{equation*}
1\geq \frac{1}{4}\gamma^{-2}\mu^2\log^2 (3N/\delta).
\end{equation*}
It is easy to see $\mu_c\geq a \overline{\sigma}$, where
$\overline{\sigma} = (\mu\mu_c^2)^{1/2}$. In this case, $\lambda
\leq 4\mu_c$. Thus, (\ref{newcond9}) holds if
\begin{equation}\label{newcond10}
1\geq 32\mu_c^2\log (\frac{2N}{\delta}).
\end{equation}
\item Case 2: $\mu\geq \mu_c^2$. The condition that
(\ref{newcond6}) holds is $a\leq 2/\mu_c$ or equivalently,
\begin{equation*}
1\geq \frac{1}{4}\gamma^{-2}\mu_c^2\log (3N/\delta).
\end{equation*}
If $\mu_c\geq a\overline{\sigma}$, where $\overline{\sigma} = \mu$,
$\lambda \leq 4\mu_c$ and the condition is again (\ref{newcond10}).
Otherwise, $\lambda\leq 4a\overline{\sigma}$. In this case,
(\ref{newcond9}) holds if
\begin{equation*}
1\geq 32\gamma^{-1}\mu^2\log (\frac{2N}{\delta}).
\end{equation*}
\end{itemize}

In conclusion, the Exact Recovery Principle is verified if $1\geq
\max(c_1\mu^2\log^2(3N/\delta), c_2\mu_c^2\log (3N/\delta))$, where
$c_1$ and $c_2$ are known positive constants.

Finally, note that $\mu^2 \leq \mathcal{O}(\frac{N\log N}{BM})$ and
$\mu_c^2 \leq \mathcal{O}(\frac{NK}{BM})$ and the assumption that
$K\geq 16\log (\frac{2N}{\delta})$, the sufficient condition for
exact recovery is $M \geq \mathcal{O}(\frac{N}{B}K\log
(\frac{N}{\delta}))$. When $\mat F$ is dense and uniform, the
condition becomes $M \geq \mathcal{O}(K\log (\frac{N}{\delta}))$.

\end{proof}

\section{Numerical Experiments}\label{numerical_experiment}
\subsection{Simulation with Sparse Signals}
In this section, we evaluate the sensing performance of several
structurally random matrices and compare it with that of the
completely random projection. We also explore the connection among
sensing performance (probability of exact recovery), streaming
capacity (block size of $\mat F$) and structure of the sparsifying
basis $\mat \Psi$ (e.g. sparsity and heterogeneity).

In the first simulation, the input signal $\vect x$ of length
$N=256$ is sparse in the DCT domain, i.e. $\vect x =\mat \Psi\vect
\alpha$, where the sparsifying basis $\mat \Psi$ is the $256\times
256$ IDCT matrix. Its transform coefficient vector $\vect \alpha$
has $K$ nonzero entries whose magnitudes are Gaussian distributed
and locations are at uniformly random, where
$K\in\{10,20,30,40,50,60\}$. With the signal $\vect x$, we generate
a measurement vector of length $M=128$: $\vect y = \mat \Phi\vect
x$, where $\mat \Phi$ is some structurally random matrix or a
completely Gaussian random matrix. SRMs under consideration are
summarized in Table ~\ref{table:srm_table1}.

\begin{table*}
\centering%
\caption{SRMs employed in the experiment with sparse signals}
\label{table:srm_table1}
\begin{tabular}{| c | c | c | c |}
\hline  Notation & \textbf{R} & \textbf{F}\\
\hline WHT64-L   & Local randomizer  & $64\times 64$ block diagonal WHT \\
\hline WHT64-G   & Global randomizer & $64\times 64$ block diagonal WHT \\
\hline WHT256-L  & Local randomizer  & $256\times 256$ block diagonal WHT \\
\hline WHT256-G  & Global randomizer & $256\times 256$ block diagonal WHT \\
\hline
\end{tabular}
\end{table*}

The software $l_1$-magic \cite{CandesRob06} is employed to recover
the signal from its measurements $\vect y$. For each value of
sparsity $K\in\{10,20,30,40,50,60\}$, we repeat the experiment 500
times and count the probability of exact recovery. The performance
curve is plotted in Fig.~\ref{syncurves}(a). Numerical values on the
$x$-axis denote signal sparsity $K$ while those on the $y$-axis
denote the probability of exact recovery. We then repeat similar
experiments when an input signal is sparse in some sparse and
non-uniform basis $\mat \Psi$. Fig.~\ref{syncurves}(b) and Fig.
~\ref{syncurves}(c) illustrate the performance curves when $\mat
\Psi$ is the Daubechies-8 wavelet basis and the identity matrix,
respectively.

There are a few notable observations from these experimental
results. First, performance of the SRM with the dense transform
matrix $\mat F$ (all of its entries are non-zero) is in average
comparable to that of the completely random matrix. Second,
performance of the SRM with the sparse transform matrix $\mat F$,
however, depends on the sparsifying basis $\mat \Psi$ of the signal.
In particular, if $\mat \Psi$ is dense, the SRM with sparse $\mat F$
also has average performance comparable with the completely random
matrix. If $\mat \Psi$ is sparse, the SRM with sparse $\mat F$ often
has worse performance the SRM with dense $\mat F$, revealing a
trade-off between sensing performance and streaming capacity. These
numerical results are consistent with the theoretical analysis
above. In addition, Fig. ~\ref{syncurves}(b) shows that the SRM with
the global randomizer seems to work much better than the SRM with
the local randomizer when the sparsifying basis $\mat \Psi$ of the
signal is sparse.

\begin{figure}[t]
\centering%
\includegraphics[width=8.5cm]{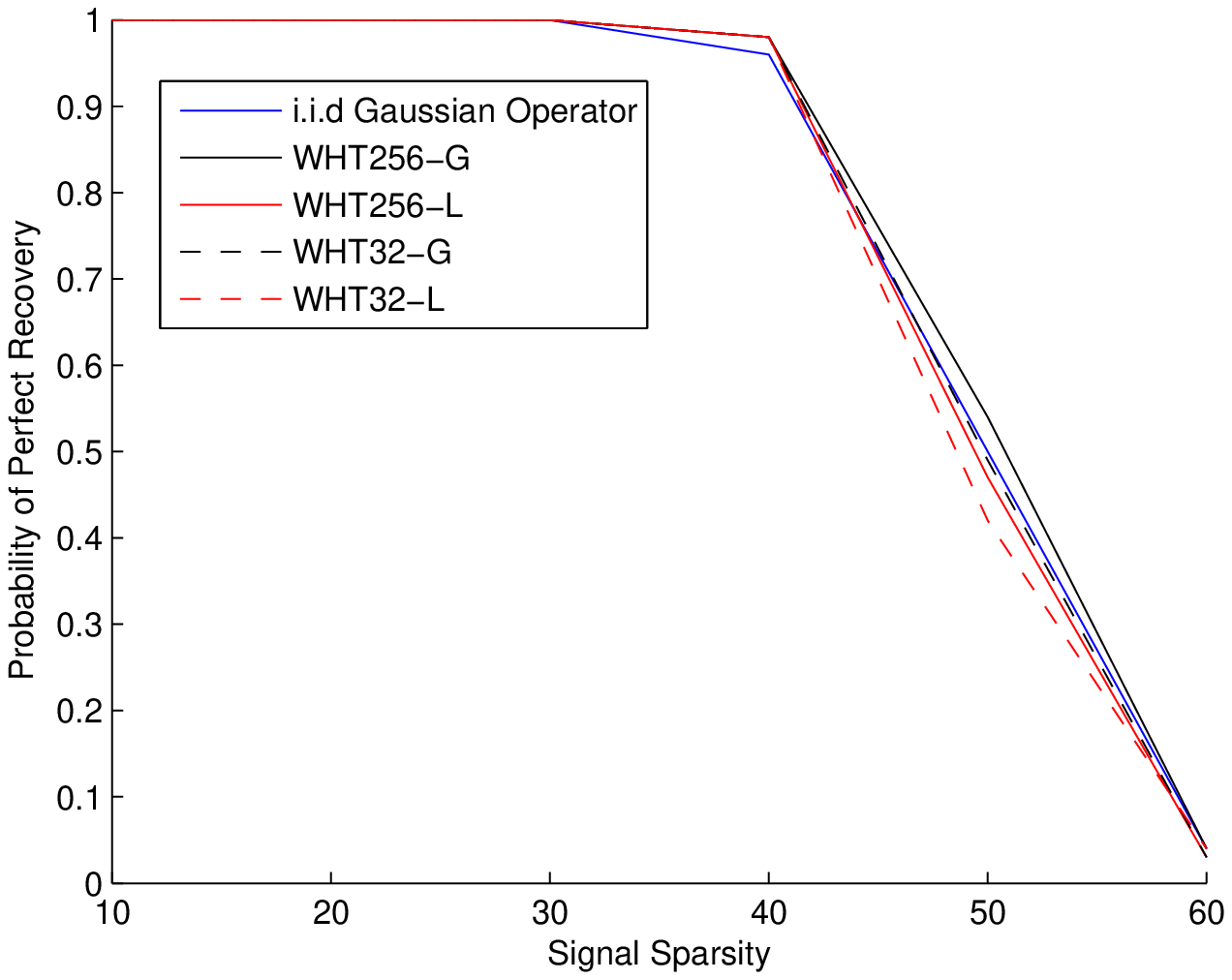}
\centerline{(a)}
\includegraphics[width=8.5cm]{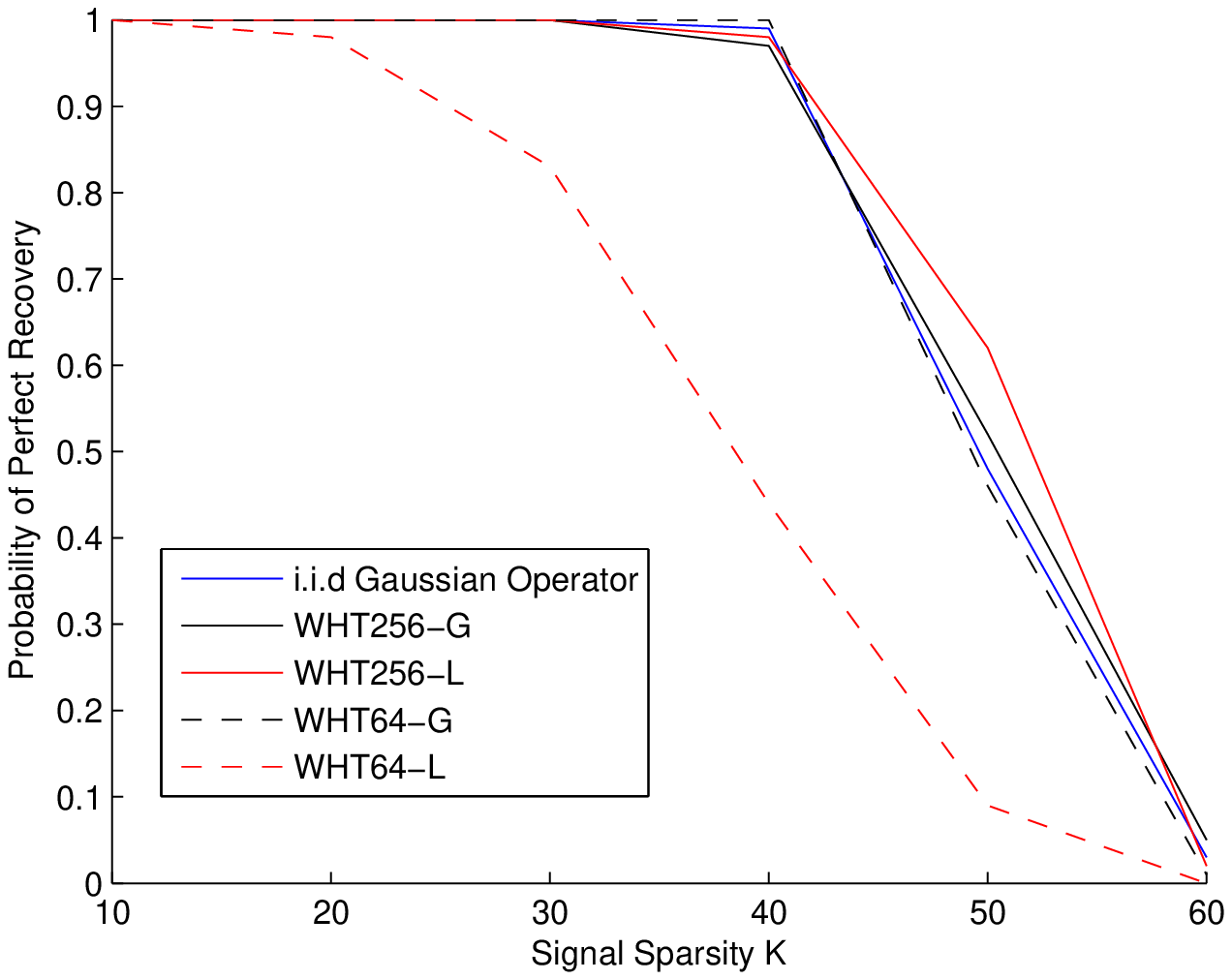}
\centerline{(b)}
\includegraphics[width=8.5cm]{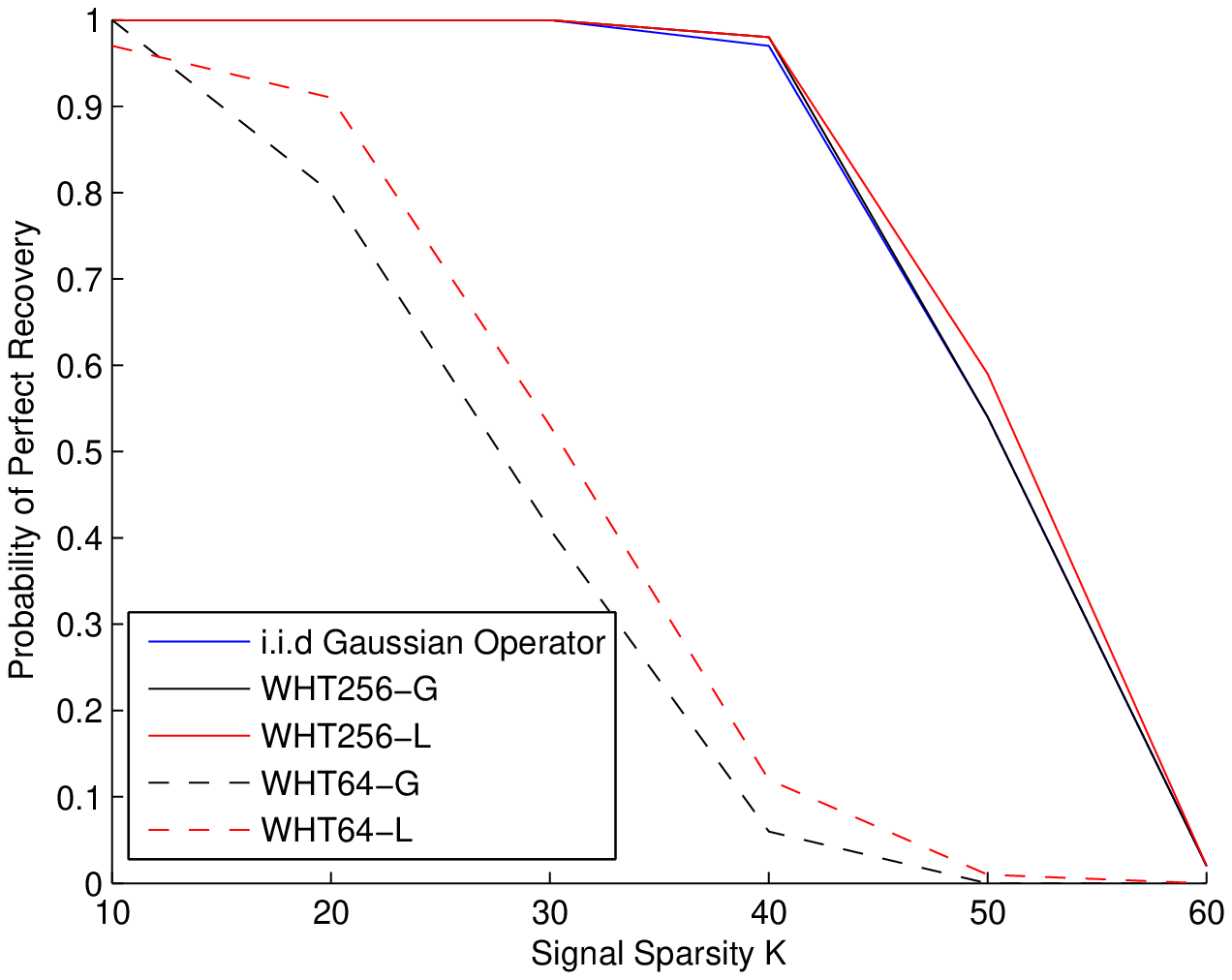}
\centerline{(c)} \caption{Performance curves: probability of exact
recovery vs. Sparsity $K$. (a) when $\mat \Psi$ is IDCT basis. (b)
when $\mat \Psi$ is Daubechies-8 wavlet basis. (c) when $\mat \Psi$
is the identity basis}\label{syncurves}
\end{figure}

\subsection{Simulation with Compressible Signals}
In this simulation, signals of interest are natural images of size
$512\times 512$ such as the $512\times 512$ Lena, Barbara and Boat
images. The sparsifying basis $\mat \Psi$ used for these natural
images is the well-known Daubechies $9/7$ wavelet transform. All
images are implicitly regarded as 1-D signals of length $512^2$. The
GPSR software in \cite{MarioGrad07} is used for signal
reconstruction.

For such a large scale simulation, it takes a huge amount of system
resources to implement the sensing method of a completely random
matrix. Thus, for the purpose of benchmark, we adopt a more
practical scheme of partial FFT in the wavelet domain (WPFFT). The
WPFFT is to sense wavelet coefficients in the wavelet domain using
the method of partial FFT. Theoretically, WPFFT has optimal
performance as the Fourier matrix is completely incoherent with the
identity matrix. The WPFFT is a method of sensing a signal in the
transform domain that also requires substantial amount of system
resources. SRMs under consideration are summarized in Table
~\ref{table:srm_table2}.

\begin{table*}
\centering%
\caption{SRMs employed in the experiment with compressible signals}
\label{table:srm_table2}
\begin{tabular}{| c | c | c | c |}
\hline  Notation & \textbf{R} & \textbf{F}\\
\hline DCT32-G   & Global randomizer  & $32\times 32$ block diagonal DCT \\
\hline WHT32-G   & Global randomizer & $32\times 32$ block diagonal WHT \\
\hline DCT512-L  & Local randomizer  & $512\times 512$ block diagonal DCT \\
\hline WHT512-L  & Local randomizer & $512\times 512$ block diagonal WHT \\
\hline
\end{tabular}
\end{table*}

For the purpose of comparison, we also implement two popular sensing
methods: partial FFT in the time domain (PFFT)\cite{CandesRob06} and
the Scrambled/Permutted FFT (SFFT) in \cite{CandesSta06,
DuarteFas05} that is equivalent to the dense SRM using the global
randomizer.

The performance curves of these sensing ensembles are plotted in
Fig.~\ref{imagecurves}(a), Fig.~\ref{imagecurves}(b) and
Fig.~\ref{imagecurves}(c), which correspond to the input signal
Lena, Barbara and Boat images, respectively. Numerical value on the
$x$-axis represents sampling rate, which is the number of
measurements over the total number of samples. Value on $y$-axis is
the quality of reconstruction (PSNR in dB). Lastly,
Fig.~\ref{visualrecons} shows the visually reconstructed $512\times
512$ Boat image from $35\%$ of measurements using WPFFT, WHT32-G and
WHT512-L ensembles.

\begin{figure}[htb]
\centering%
\includegraphics[width=8.5cm]{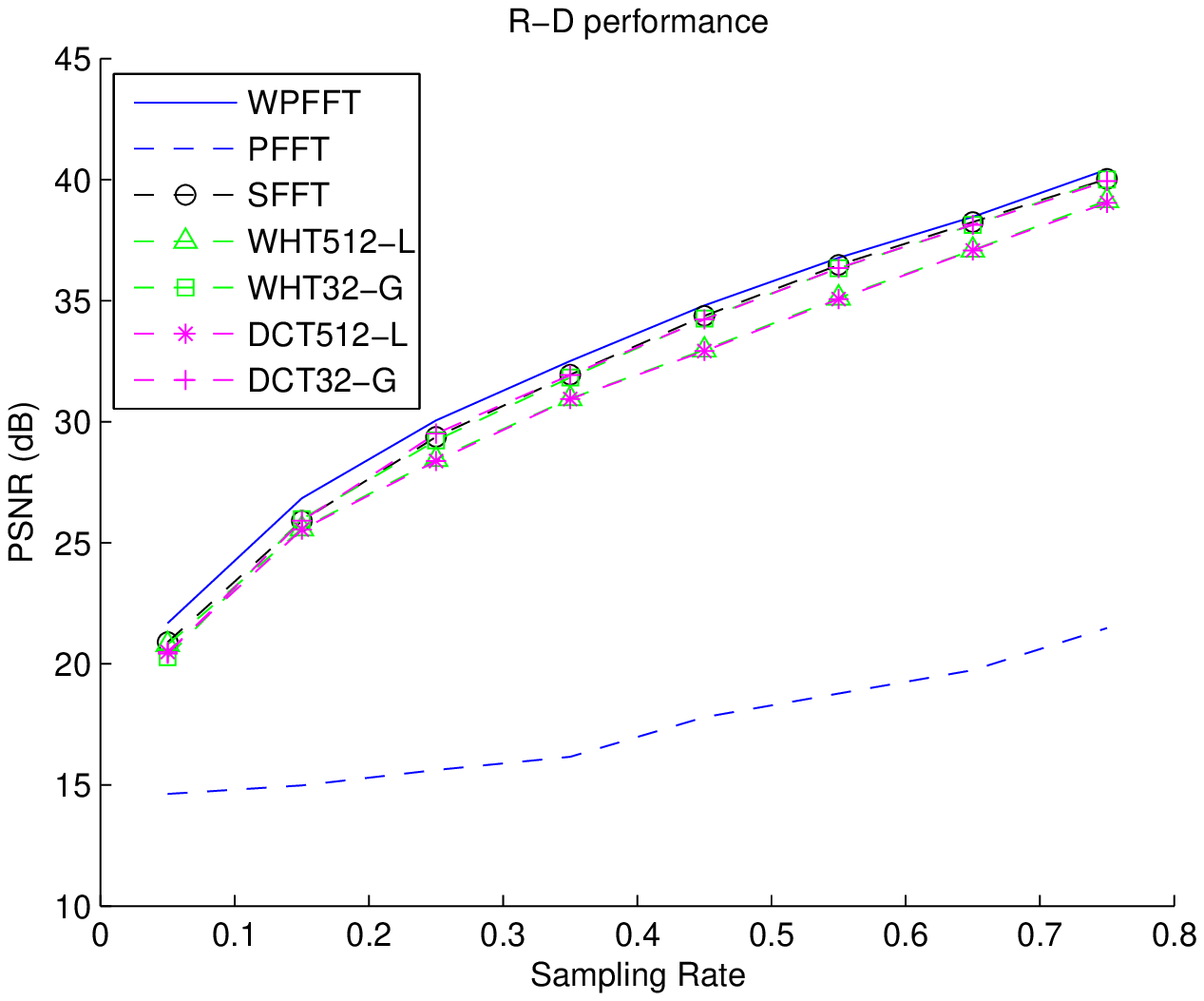}
\centerline{(a)}
\includegraphics[width=8.5cm]{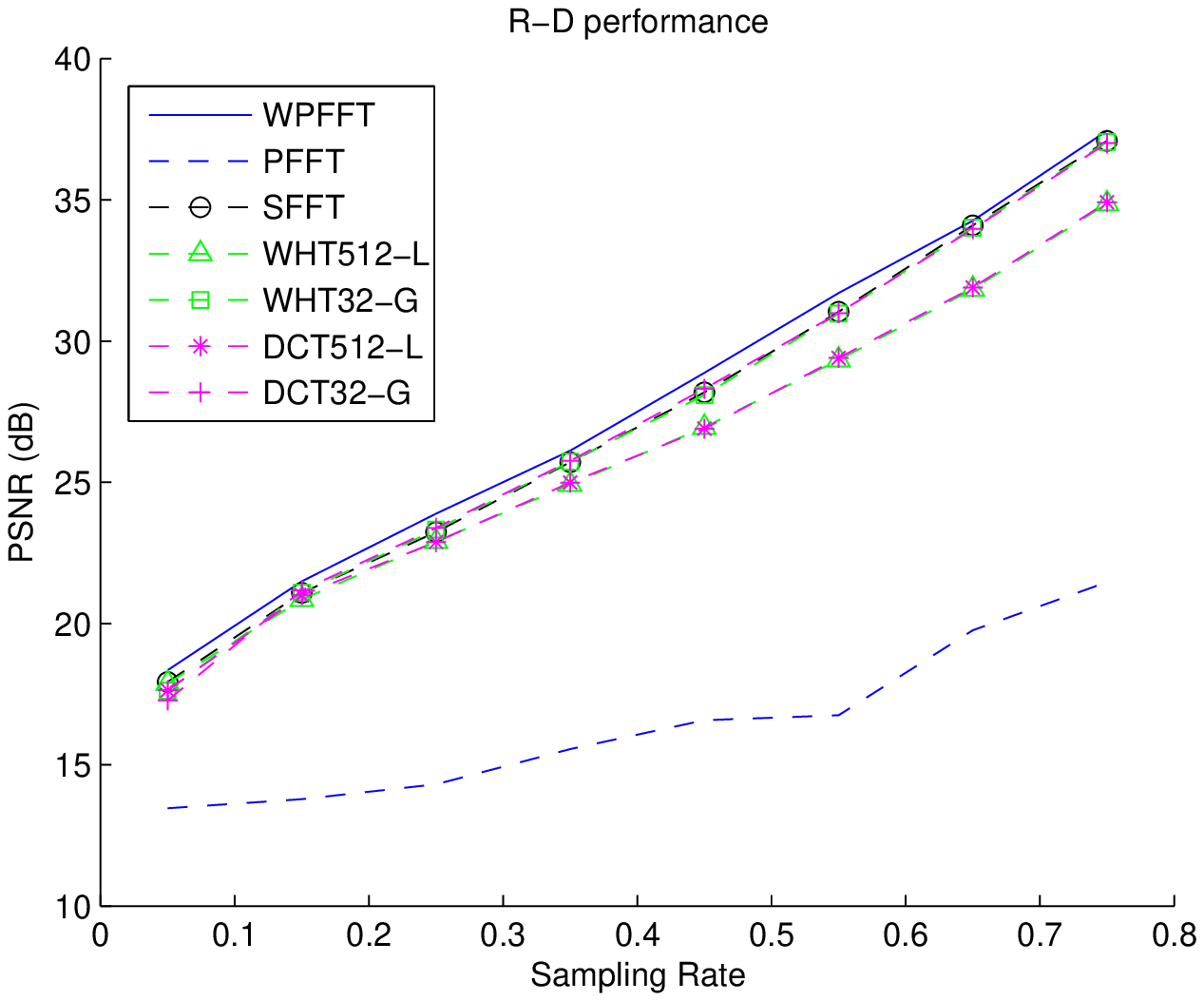}
\centerline{(b)}
\includegraphics[width=8.5cm]{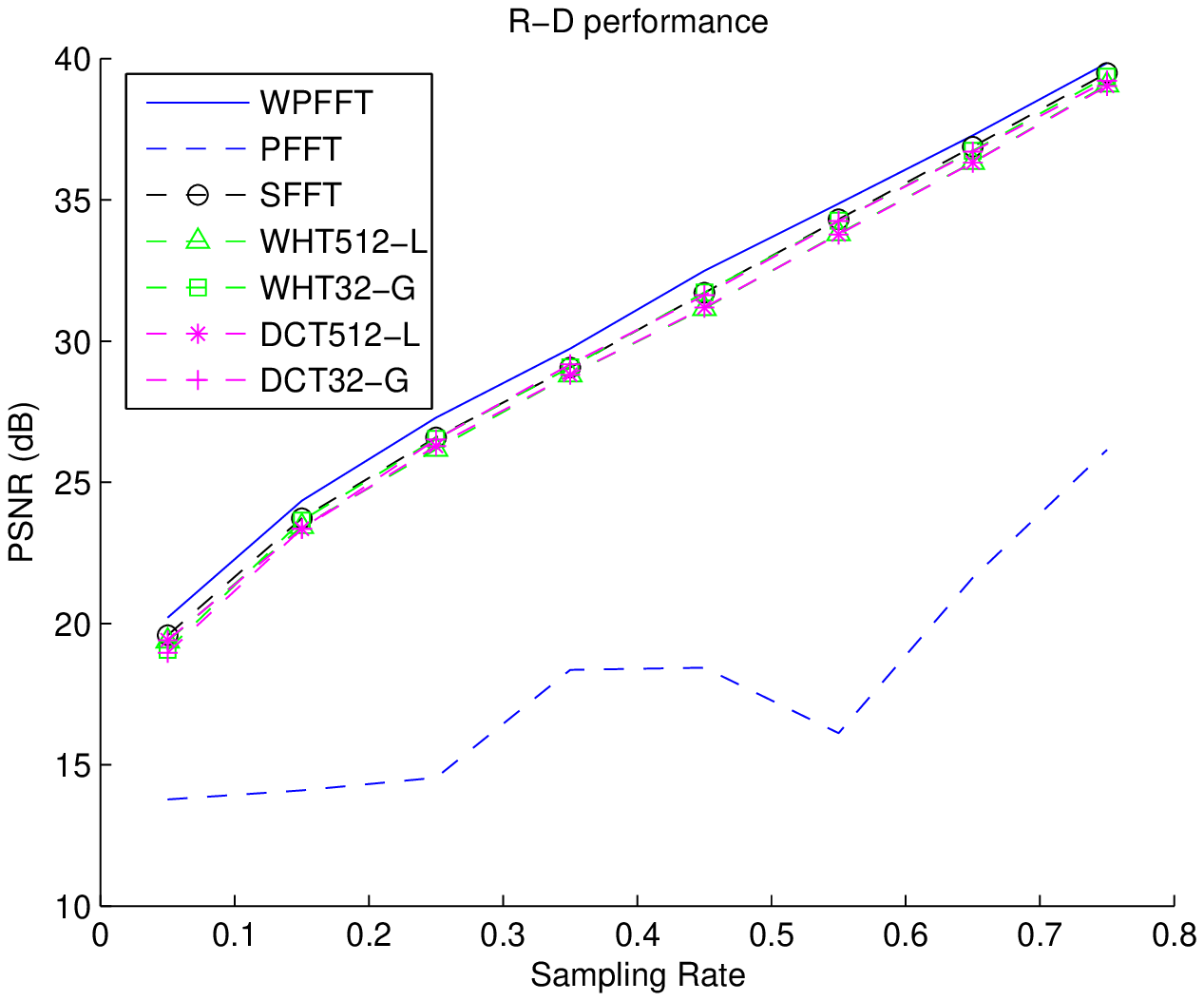}
\centerline{(c)} \caption{Performance curves: Quality of signal
reconstruction vs. sampling rate $M/N$. (a) the $512\times 512$ Lena
image. (b) the $512\times 512$ Barbara image. (c) the $512\times
512$ Boat image}\label{imagecurves}
\end{figure}

As clearly seen in Fig. 3, the PFFT is not an efficient sensing
matrix for smooth signals like images because Fourier matrix and
wavelet basis are highly coherent. On the other hand, the SRM
method, which can roughly be viewed as the PFFT preceded by the
pre-randomization process, is very efficient. In particular, with a
dense SRM like SFFT, the performance difference between the SRM
method and the benchmark one, WPFFT, is less than $1$ dB. In
addition, performance of DCT512-L and WHT512-L that are fully
streaming capable SRM, degrades about $1.5$ dB, which is a
reasonable sacrifice as the buffer size required is less than $0.2$
percent of the total length of the original signal. Less degradation
is obtainable when the buffer size is increased. Also, in all cases,
there is no observable difference of performance between DCT and
normalized WHT transforms. It implies that orthonormal matrices
whose entries have the same order of absolute magnitude generate
comparable performance.  In addition, highly sparse SRM using the
global randomizer such as DCT32-G and WHT32-G has experimental
performance comparable to that of the dense SRMs. Note that these
SRM are highly sparse because their density are only $2^{-13}$. This
observation again verifies that SRM with the global randomizer
outperforms SRM with the local randomizer. This might indicate that
our theoretical analysis for the global randomizer is inadequate. In
practice, we believe that the global randomizer always works as well
as and even better than the local randomizer. We leave the
theoretical justification of this observation for our future
research.

\begin{figure*}[hp]
\vspace{-1cm}
  \begin{center}
   \vspace{-0.2cm}
   \subfigure[]{\label{fig:edge-a}\includegraphics[width=6.0cm]{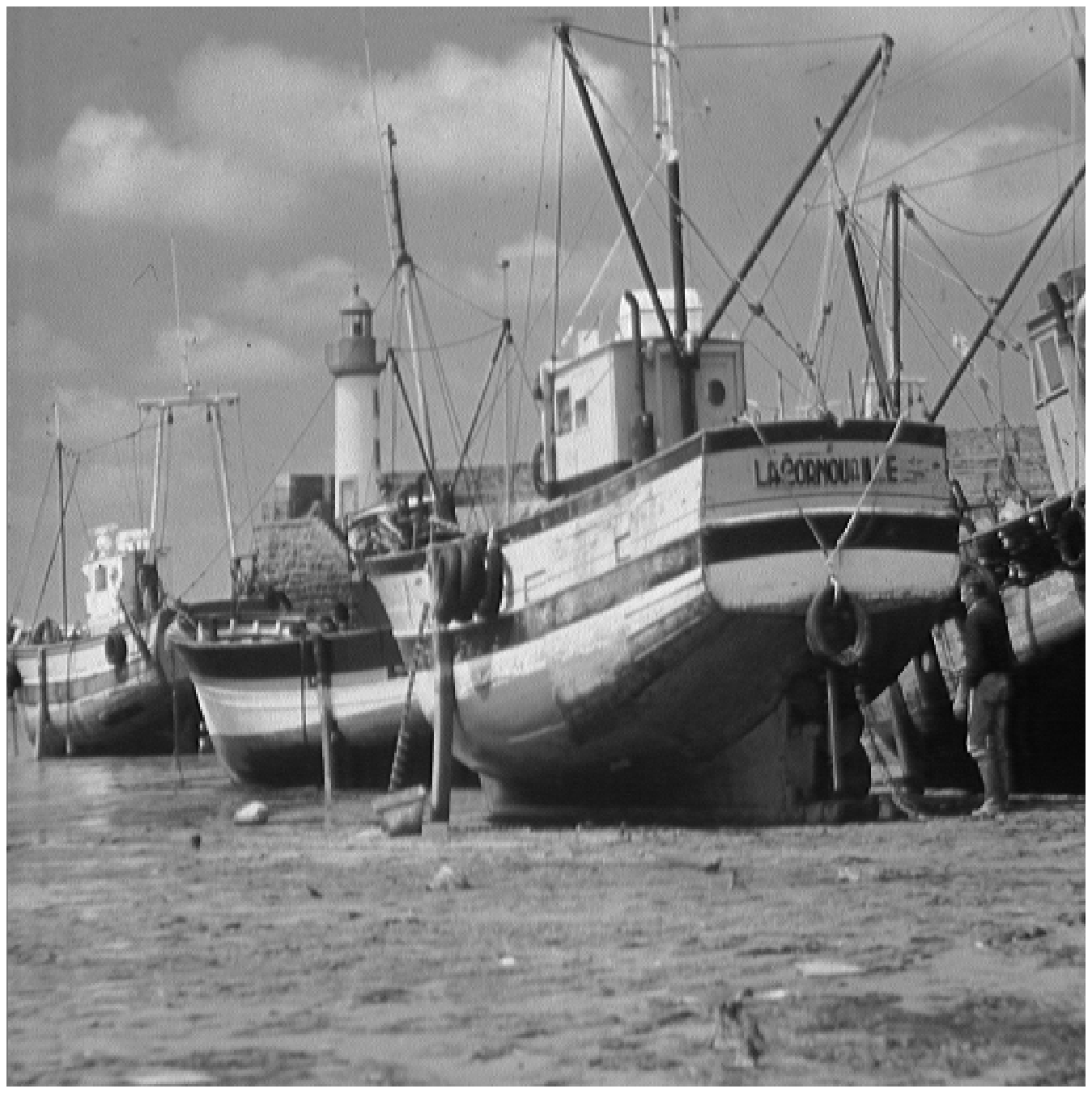}}
   \hspace{0.5cm}
    \subfigure[]{\label{fig:edge-b}\includegraphics[width=6.0cm]{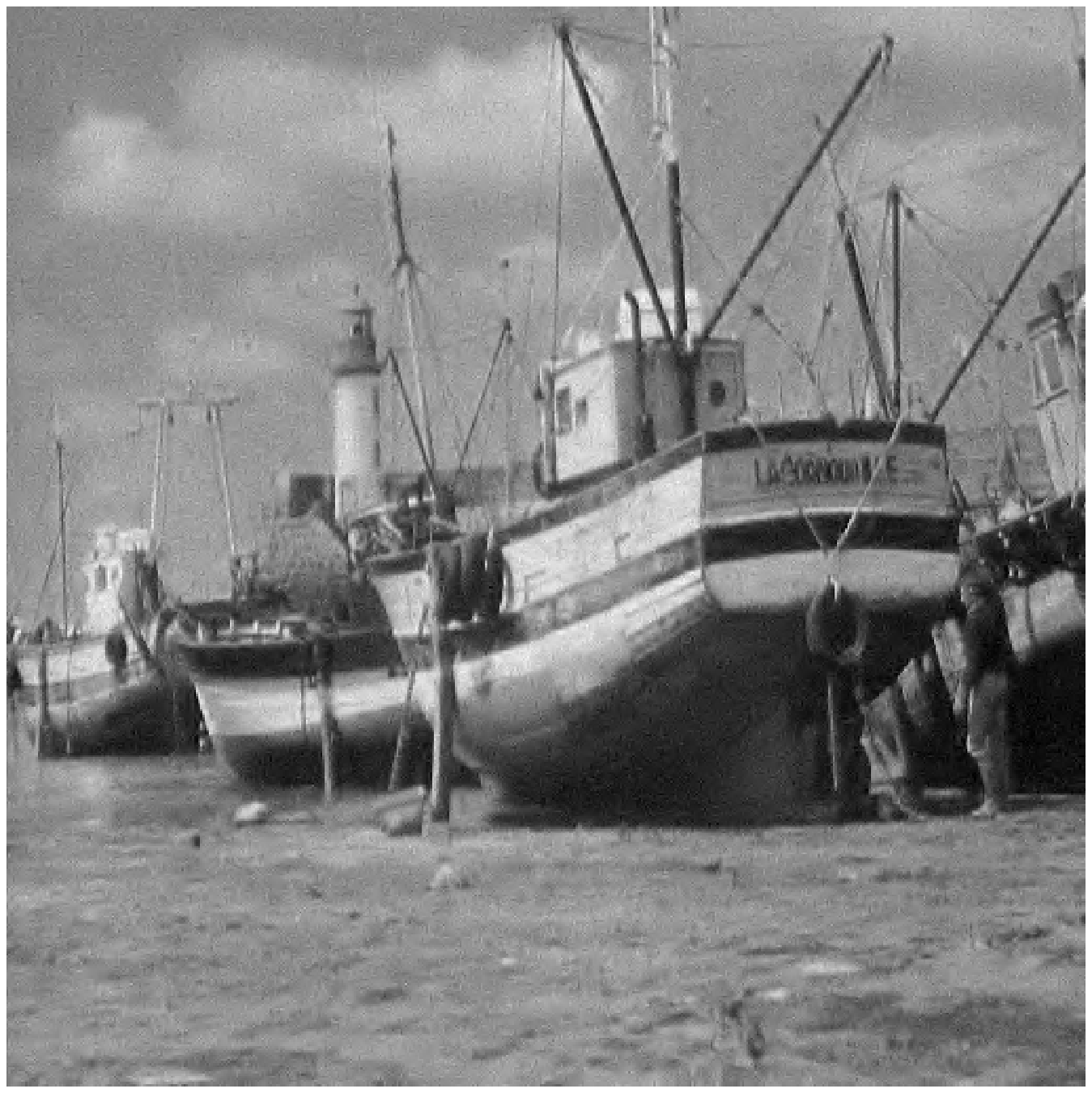}}
    \end{center}
      \begin{center}
      \vspace{-0.2cm}
    \subfigure[]{\label{fig:edge-c}\includegraphics[width=6.0cm]{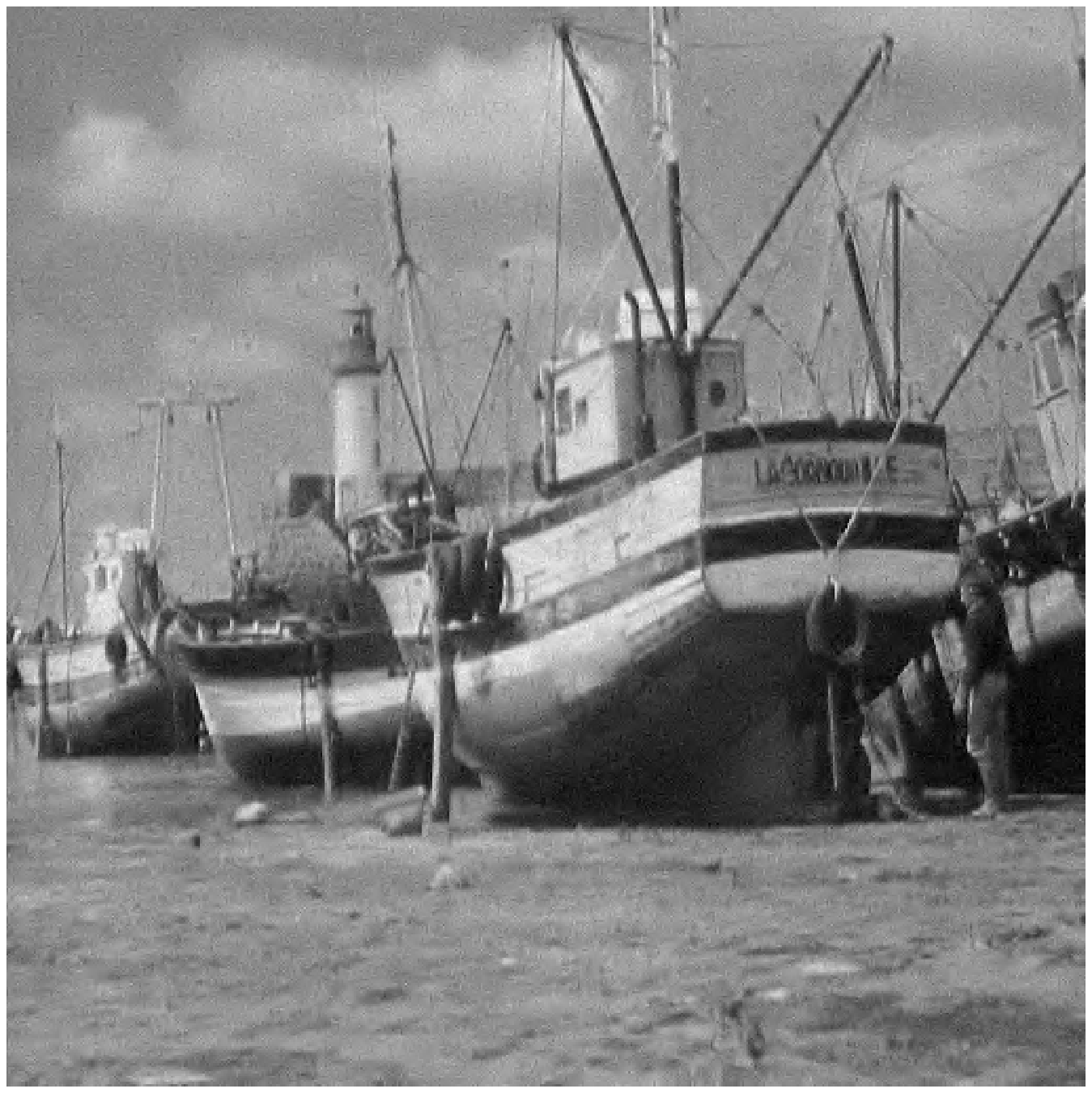}}
    \hspace{0.5cm}
     \subfigure[]{\label{fig:edge-d}\includegraphics[width=6.0cm]{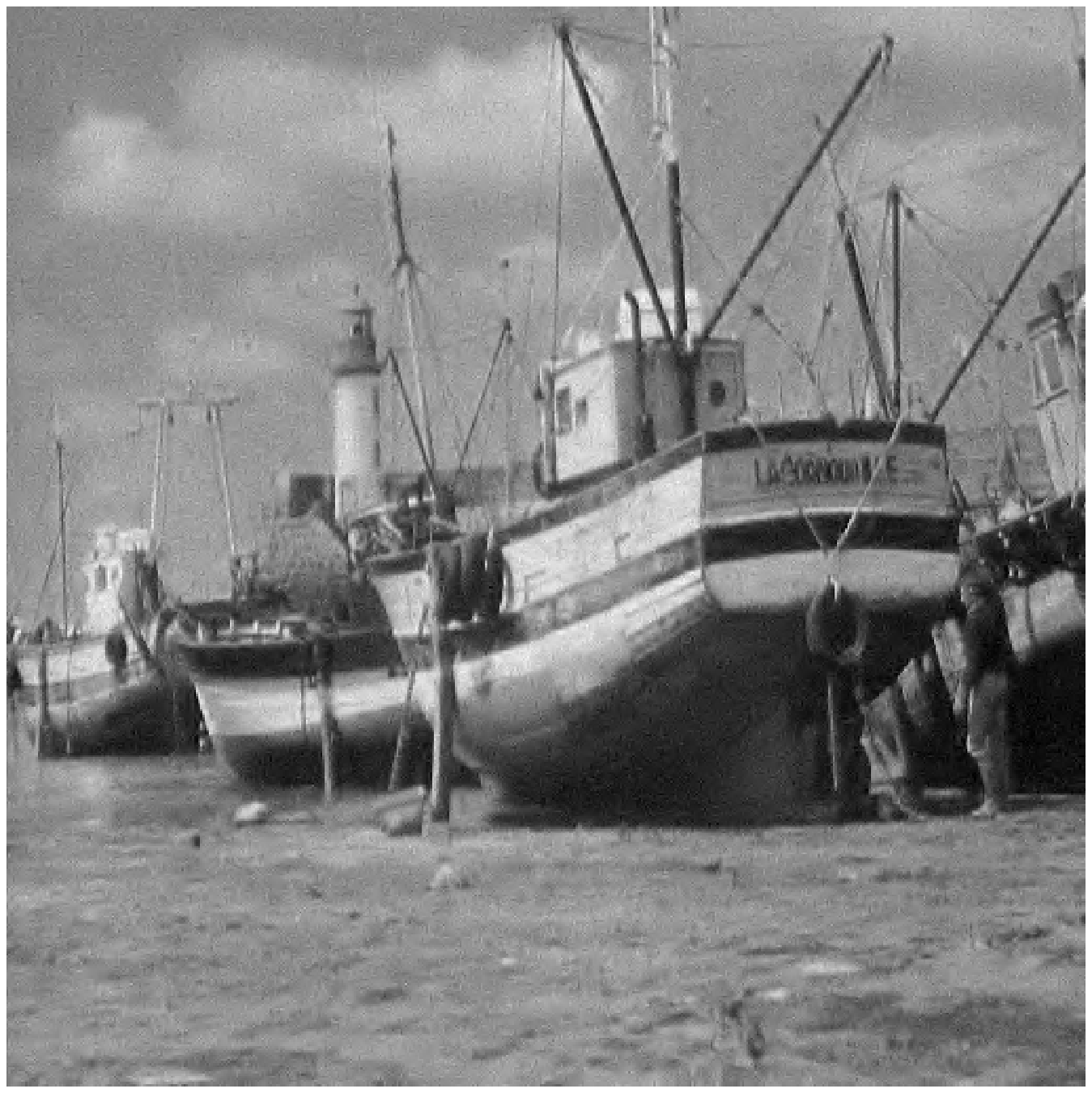}}
    \end{center}
  \caption{Reconstructed $512\times512$ \emph{Boat} images from $M/N=35\%$
   sampling rate. (a) The original Boat image; (b) using the WPFFT ensemble: 28.5dB; (c) using the WHT32-G ensemble: 28dB; (d) using the WHT512-L ensemble: 27.7dB}
  \label{visualrecons}
\end{figure*}

\section{Discussion and Conclusion}\label{discussion_conclusion}

\subsection{Complexity Discussion}
We compare the computation and memory complexity between the
proposed SRM and other random sensing matrices such as Gaussian or
Bernoulli i.i.d. matrices. In implementation, the i.i.d Bernoulli
matrix is obviously preferred than i.i.d Gaussian one as the former
has integer entries $\{1,-1\}$ and requires only 1 bit to represent
each entry. A $M\times N$ i.i.d. Bernoulli sensing matrix requires
$MN$ bits for storing the matrix and $MN$ additions and
multiplications for sensing operation. An $M\times N$ SRM only
requires $2N + N\log N$ bits for storage and $N + N\log N$ additions
and multiplications for sensing operation. With the SRM method, the
computational complexity and memory space required is independent
with the number of measurements $M$. Note that with the SRM method,
we do not need to store matrices $\mat D$, $\mat F$, $\mat R$
explicitly. We only need to store the diagonals of $\mat D$ and of
$\mat R$ and the fast transform $\mat F$, resulting in significant
saving of both memory space and computational complexity.

Computational complexity and running time of $l_1$-minimization
based reconstruction algorithms often depend critically on whether
matrix-vector multiplications $\mat A\vect u$ and $\mat A^T\vect u$
can be computed quickly and efficiently (where $\mat A = \mat
\Phi\mat \Psi$) \cite{MarioGrad07}. For the sake of simplicity,
assuming that $\mat \Psi$ is identity matrix. $\mat A\vect u = \mat
\Phi\vect u$ requires $MN = \mathcal{O}(KN\log N)$ additions and
multiplications for a random sensing matrix $\mat \Phi$ and
$\mathcal{O}(N\log N)$ additions and multiplications for the SRM
method. This implies that at each iteration, SRM can speed up the
reconstruction algorithm with at least $K$ folds. With compressible
signals (e.g., images), the number of measurements acquired tends to
be proportional with the signal dimension, for example, $M = N/4$.
In this case, using SRM can achieve computational complexity
reduction with the factor of $\frac{N}{4\log N}$ times.

Table \ref{table:feature_comparison} summarizes computational
complexity and practical advantages between SRM and a random sensing
matrix.

\begin{table*}
\centering%
\caption{Practical feature comparison}
\label{table:feature_comparison}
\begin{tabular}{| c | c | c | c |}
\hline Features                                   & SRMs                             & Completely Random Matrices\\
\hline No. of measurements for exact recovery     & $\mathcal{O}(K\log N)$           & $\mathcal{O}(K\log N)$ \\
\hline Sensing complexity                         & $N\log N$                        & $\mathcal{O}(KN\log N)$ \\
\hline Reconstruction complexity at each iteration& $\mathcal{O}(N\log N)$           & $\mathcal{O}(KN\log N)$ \\
\hline Fast computability                         & Yes                              & No \\
\hline Block-based processing                     & Yes                              & No \\
\hline
\end{tabular}
\end{table*}

\subsection{Relationship with Other Related Works}

When $\mat R$ is the local randomizer, SRM is a little reminiscent
to the so-called Fast Johnson-Lindenstrauss Transform (FJLT)
\cite{NirFJLT06}. However, SRM employs a simpler matrix $\mat D$. In
FJLT, this matrix $\mat D$ is a completely random matrix with sparse
distribution. It is unknown if there exists an efficient
implementation of such a sparse random matrix. SRM is relevant for
practical applications because of its high performance and fast
computation.

In \cite{CandesSta06, DuarteFas05}, the Scrambled/Permuted FFT is
experimentally proposed as a heuristic low-complexity sensing method
that is efficient for sensing a large signal. To the best of our
knowledge, however, there has not been any theoretical analysis for
the Scrambled FFT. SRM is a generalized framework in which Scrambled
FFT is just a specific case, and thus verifying the theoretical
validity of the Scrambled FFT.

Random Convolution convolving the input signal with a random pulse
followed by randomly subsampling measurements is proposed in
\cite{JustinRandomconvol08} as a promising sensing method for
practical applications. Although there are a few other methods that
exploit the same idea of convolving a signal with a random pulse,
for examples: Random Filter in \cite{TroppRan06} and Toeplitz
structured sensing matrix in \cite{WaheedToeplitzCS07}, only the
Random Convolution method can be shown to approach optimal sensing
performance. While sensing methods such as Random Filter and
Toeplitz-based CS methods subsample measurements structurally, the
Random Convolution method subsamples measurements in a random
fashion, a technique that is also employed in SRM. In addition, the
Random Convolution method introduces randomness into the Fourier
domain by randomizing phases of Fourier coefficients. These two
techniques decouple stochastic dependence among measurements and
thus, giving the Random Convolution method a higher performance.

SRM is distinct from all aforementioned methods, including the
Random Convolution one. A key difference is that SRM pre-randomizes
a sensing signal directly in its original domain (via the global
randomizer or the local randomizer) while the Random Convolution
method pre-randomizes a sensing signal in the Fourier domain.  SRM
also extends the Random Convolution method by showing that not only
Fourier transform but also other popular fast transforms, such as
DCT or WHT, can be employed to achieve similar high performance. In
conclusion, among existing sensing methods, the SRM framework
presents an alternative approach to design high performance,
low-complexity sensing matrices with practical and flexible
features.


%
%
\appendices
\section{}\label{appdx1}
\newtheoremstyle{nonum}{}{}{\itshape}{}{\bfseries}{.}{ }{\thmname{#1}\thmnote{ (\mdseries #3)}}
\theoremstyle{nonum}
\newtheorem{CLT}{Central Limit Theorem}
\begin{CLT} Let $Z_1,Z_2,\dots,Z_N$ be mutually independent random
variables. Assume $E(Z_k)=0$ and denote $\sigma^2 = \sum_{k=1}^N
\text{Var}(Z_k)$ . If for a given $\epsilon\geq 0$ and $N$
sufficiently large, the following inequalities hold:
\begin{equation*}
\text{Var}(Z_k)< \epsilon \sigma^2 \hspace{0.4cm} k=1,2,\dots,N
\end{equation*}
then distribution of the normalized sum $S = \sum_{k=1}^N Z_k$
converges to $\mathcal{N}(0,\sigma^2)$
\end{CLT}

\theoremstyle{nonum}
\newtheorem{ComCLT}{Combinatorial Central Limit Theorem}
\begin{ComCLT}Given two sequences
$\{a_k\}_{k=1}^N$ and $\{b_k\}_{k=1}^N$. Assume the $a_k$ are not
all equal and $b_k$ are also not all equal. Let $[\omega_1,
\omega_2,\dots,\omega_N]$ be a uniform random permutation of
$[1,2,...,N]$. Denote $Z_k = a_{\omega_k}$ and
\begin{equation*}
S = \sum_{k=1}^N Z_k b_k ;
\end{equation*}
$S$ is asymptotically normally distributed
$\mathcal{N}(E(S),\text{Var}(S))$ if
\begin{equation*}
\lim_{N\rightarrow \infty} N\frac{\max_{1\leq k\leq
N}(Z_k-\overline{Z})^2}{\sum_{k=1}^N
(Z_{k}-\overline{Z})^2}\frac{\max_{1\leq k\leq N}
(b_k-\overline{b})^2}{\sum_{k=1}^N (b_k-\overline{b})^2}  = 0;
\end{equation*}
where
\begin{equation*}
\overline{b} = \frac{1}{N}\sum_{k=1}^N
b_k\hspace{0.4cm}\text{and}\hspace{0.4cm}\overline{Z} =
\frac{1}{N}\sum_{k=1}^N Z_k.
\end{equation*}
\end{ComCLT}

\section{}\label{appdx2}
\theoremstyle{nonum}
\newtheorem{hoeff}{Hoeffding's Concentration Inequality}
\begin{hoeff}
Suppose $X_1, X_2,...,X_N$ are independent random variables and
$a_k\leq X_K\leq b_k$ ($k=1,2,...,N$). Define a new random variable
$S = \sum_{k=1}^N X_k$. Then for any $t>0$
\begin{equation*}
P(|S-E(S)|\geq t)\leq 2e^{-\frac{2t^2}{\sum_{k=1}^N (b_k-a_k)^2}}.
\end{equation*}
\end{hoeff}

\theoremstyle{nonum}
\newtheorem{ledoux}{Ledoux's Concentration Inequality}
\begin{ledoux}
Let $\{ \eta_i \}_{1 \leq i \leq N}$ be a sequence of independent
random variables such that $|\eta_i| \leq 1$ almost surely and
$\vect v_1$, $\vect v_2$,\dots, $\vect v_N$ be vectors in Banach
space. Define a new random variable: $S = \|\sum_{i=1}^N \eta_i
\vect v_i\|$. Then for any $t > 0$,
\begin{equation*}
P(S\geq E(S)+t)\leq 2\exp(-\frac{t^2}{16\sigma^2})
\end{equation*}
where $\sigma^2$ denote the variance of $S$ and $\sigma^2 =
\sup_{\|\vect u\|\leq 1}\sum_{i=1}^N |\langle \vect u,\vect
v_i\rangle|^2$.
\end{ledoux}

\theoremstyle{nonum}
\newtheorem{Talagrand}{Talagrand's Concentration Inequality}
\begin{Talagrand}
Let $Z_k$ be zero-mean i.i.d random variables and bounded $|Z_k|\leq
\lambda$ and $\vect u_k$ be column vectors of a matrix $\mat U$.
Define a new random variable: $S = \|\sum_{i=1}^N Z_k \vect u_k\|$.
Then for any $t>0$:
\begin{equation*}
P(S\geq E(S)+t)\leq 3\exp(-\frac{t}{cB}\log (1+\frac{Bt}{\sigma^2 +
BE(S)}))
\end{equation*}
where $c$ is some constant, variance $\sigma^2 = E(Z_k^2)\|\mat
U\|^2$ and $B = \lambda \max_{1\leq k\leq N}\|\vect u_k\|$.
\end{Talagrand}

\theoremstyle{nonum}
\newtheorem{Sourav}{Sourav's Concentration Inequality}
\begin{Sourav}
Let $\{Z_{ij}\}_{1\leq i,j\leq N}$ be a collection of numbers from
$[0,1]$. Let $[\omega_1,\omega_2,\dots,\omega_N]$ be a uniformly
random permutation of $[1,2,\dots,N]$. Define a new random variable:
$S=\sum_{i=1}^N Z_{i\omega_i}$. Then for any $t\geq 0$
\begin{equation*}
P(|S-E(S)|\geq t)\leq 2\exp(-\frac{t^2}{4E(S)+2t}).
\end{equation*}
\end{Sourav}




%
%
%
\bibliographystyle{IEEEbib}
\bibliography{research_proposal_references}

\begin{thebibliography}{10}

\bibitem{CandesRob06}
E.~Cand\`{e}s, J.~Romberg, and T.~Tao,
\newblock ``Robust uncertainty principles: Exact signal reconstruction from
  highly incomplete frequency information,''
\newblock {\em IEEE Trans. Inf. Theory}, vol. 52, pp. 489 -- 509, 2006.

\bibitem{DonohoCom06}
D.~L. Donoho,
\newblock ``Compressed sensing,''
\newblock {\em IEEE Trans. Inf. Theory}, vol. 52, no. 4, pp. 1289 -- 1306,
  2006.

\bibitem{MarioGrad07}
M.~A.~T. Figueiredo, R.~D. Nowak, and S.~J. Wright,
\newblock ``Gradient projection for sparse reconstruction,''
\newblock {\em IEEE J. Sel. Topics Signal Process.}, vol. 1, no. 4, pp.
  586--597, 2007.

\bibitem{ElaineFix07}
E.~T. Hale, W.~Yin, and Y.~Zhang,
\newblock ``Fixed-point continuation for l-minimization: Methodology and
  convergence,''
\newblock {\em SIAM J. Opt.}, vol. 19, no. 3, pp. 1107--1130, 2008.

\bibitem{EwoutProbing08}
E.~V.~D. Berg and M.~P. Friedlander,
\newblock ``Probing the pareto frontier for basis purusit solutions,''
\newblock {\em SIAM J. Scien. Comp.}, vol. 31, no. 2, pp. 890--912, 2008.

\bibitem{TroppSig05}
J.~Tropp and A.~Gilbert,
\newblock ``Signal recovery from random measurements via orthogonal matching
  pursuit,''
\newblock {\em IEEE Trans. Info. Theory}, vol. 53, no. 12, pp. 4655--4666, Dec
  2007.

\bibitem{NeedellCosamp08}
D.~Needell and J.~A. Tropp,
\newblock ``Cosamp: Iterative signal recovery from incomplete and inaccurate
  samples,''
\newblock {\em Appl. Comput. Harmon. Anal.}, vol. 26, pp. 301--321, 2008.

\bibitem{WeiSub08}
W.~Dai and O.~Milenkovic,
\newblock ``Subspace pursuit for compressive sensing signal reconstruction,''
\newblock {\em IEEE Trans. Inf. Theory}, vol. 55, no. 5, pp. 2230--2249, 2009.

\bibitem{DonohoSpa06}
D.~L. Donoho, Y.~Tsaig, and J.-L. Starck,
\newblock ``Sparse solution of underdetermined linear equations by stagewise
  orthogonal matching pursuit,''
\newblock {\em Technical Report}, 2006.

\bibitem{DoSparity08}
T.~T. Do, L.~Gan, N.~Nguyen, and T.~D. Tran,
\newblock ``Sparsity adaptive matching pursuit algorithm for practical
  compressed sensing,''
\newblock {\em Asilomar Conf. Sign. Sys. Comput.}, pp. 581--587, 2008.

\bibitem{DonohoUnc01}
D.~L. Donoho and X.~Huo,
\newblock ``Uncertainty principles and ideal atomic decomposition,''
\newblock {\em IEEE Trans. Inf. Theory}, vol. 47, no. 7, pp. 2845 -- 2862,
  2001.

\bibitem{CandesNea06}
E.~Cand\`{e}s and T.~Tao,
\newblock ``Near optimal signal recovery from random projections: Universal
  encoding strategies?,''
\newblock {\em IEEE Trans. Inf. Theory}, vol. 52, no. 12, pp. 5406 -- 5425,
  2006.

\bibitem{MedelsonUni06}
S.~Mendelson, A.~Pajor, and N.~Tomczak-Jaegermann,
\newblock ``Uniform uncertainty principle for bernoulli and subgaussian
  ensembles,''
\newblock {\em Constructive Alg.}, vol. 28, pp. 269--283, 2008.

\bibitem{CandesSpa07}
E.~Cand\`{e}s and J.~Romberg,
\newblock ``Sparsity and incoherence in compressive sampling,''
\newblock {\em Inverse Problems}, vol. 23, no. 3, 2007.

\bibitem{CofimanNoi05}
R.~Coifman, F.~Geshwind, and Y.~Meyer,
\newblock ``Noiselets,''
\newblock {\em Appl. Comput. Harmon. Anal.}, vol. 10, pp. 27--44, 2001 2005.

\bibitem{CandesDec05}
E.~Cand\`{e}s and T.~Tao,
\newblock ``Decoding by linear programming,''
\newblock {\em IEEE Trans. Inf. Theory}, vol. 51, no. 12, pp. 4203--4215, 2005.

\bibitem{TroppRan06}
J.~Tropp, M.~Wakin, M.~Duarte, D.~Baron, and R.~Baraniuk,
\newblock ``Random filters for compressive sampling and reconstruction,''
\newblock {\em IEEE Conf. Acous. Speech Sign. Proc.}, vol. 3, pp. 872--875,
  2006.

\bibitem{WaheedToeplitzCS07}
W.~Bajwa, J.~Haupt, G.~Raz, S.~Wright, and R.~Nowak,
\newblock ``Toeplitz-structured compressed sensing matrices,''
\newblock {\em IEEE Stat. Sign. Proc. (SSP)}, pp. 26--29, 2007.

\bibitem{JustinRandomconvol08}
J.~Romberg,
\newblock ``Compressive sensing by random convolution,''
\newblock {\em SIAM J. Imaging Sci.}, vol. 2, pp. 1098--1128, 2009.

\bibitem{HoeffdingCom51}
W.~Hoeffding,
\newblock ``A combinatorial central limit theorem,''
\newblock {\em The Annals Math. Stat.}, vol. 22, pp. 558--566, 1951.

\bibitem{SchnassAve07}
K.~Schnass and P.~Vandergheynst,
\newblock ``Average performance analysis for thresholding,''
\newblock {\em IEEE Sign. Proc. Letters}, vol. 14, no. 11, 2007.

\bibitem{LedouxConcentration01}
M.~Ledoux,
\newblock ``The concentration of measure phenomenon,''
\newblock {\em American Mathematical Society}, 2001.

\bibitem{SouravStein07}
S.~Chatterjee,
\newblock ``Stein's method for concentration inequalities,''
\newblock {\em Probab. Theory Related Fields}, vol. 138, pp. 305--321, 2007.

\bibitem{TalagrandConcentration96}
M.~Talagrand,
\newblock ``New concentration inequalities in product spaces,''
\newblock {\em Invent. Math.}, vol. 126, pp. 505--563, 1996.

\bibitem{CandesSta06}
E.~Cand\`{e}s, J.~Romberg, and T.~Tao,
\newblock ``Stable signal recovery from incomplete and inaccurate
  measurements,''
\newblock {\em Comm. Pure Applied Math.}, vol. 59, no. 8, 2006.

\bibitem{DuarteFas05}
M.~F. Duarte, M.~B. Wakin, and R.~G. Baraniuk,
\newblock ``Fast reconstruction of piecewise smooth signals from incoherent
  projections,''
\newblock {\em Workshop Sign. Proc. Adapt. Sparse Struc. Represent.}, 2005.

\bibitem{NirFJLT06}
N.~Ailon and B.~Chazelle,
\newblock ``Approximate nearest neighbors and the fast johnson-lindenstrauss
  transform,''
\newblock {\em Proc. 38th ACM Symp. Theory Comput.}, vol. 66, pp. 557 -- 563,
  2006.

\end{thebibliography}

\end{document}